\documentclass[
 10pt,
 superscriptaddress,
 nofootinbib,
 amsmath,amssymb,
 aps,
 pra,
 notitlepage,
twocolumn
]{revtex4-1}
\usepackage{amssymb}
\usepackage{complexity}
\usepackage[utf8]{inputenc} 
\usepackage[T1]{fontenc}    
\usepackage{hyperref}       
\usepackage{url}            
\usepackage{booktabs}       
\usepackage{amsfonts}       
\usepackage{nicefrac}       
\usepackage{microtype}      
\usepackage{lipsum}
\usepackage{amsmath}
\usepackage{bbm}
\usepackage{complexity}
\usepackage{qcircuit}
\usepackage{bm}
\usepackage{physics}
\usepackage{placeins}
\usepackage{tikz}
\usepackage{forest}
\usepackage{multirow}
\usetikzlibrary{calc,arrows.meta,positioning, shapes.geometric, arrows, shapes.symbols,shapes.callouts,patterns}

\tikzset{
    every node/.style={font=\sffamily\small},
    main node/.style={thick,circle ,draw},
    visible node/.style={thick,rectangle ,draw}
}
\usepackage{subcaption}
\usepackage{qcircuit}
\usepackage{graphicx}
\usepackage{graphics}
\usepackage{amsthm}
\usepackage{latexsym}
\usepackage{mathtools}
\usepackage{upgreek}
\usepackage{textgreek}
\usepackage{thmtools}
\usepackage{thm-restate}
\usepackage{cleveref}

\newtheorem{theorem}{Theorem}
\newtheorem{corollary}{Corollary}

\newtheorem{definition}{Definition}
\newtheorem{lemma}{Lemma}

\begin{document}

\title{Generating Approximate Ground States of Molecules Using Quantum Machine Learning}

\author{Jack Ceroni}
\affiliation{Xanadu, Toronto, ON, M5G 2C8, Canada}
\affiliation{Department of Mathematics, University of Toronto, Toronto, ON, M5S 3E1, Canada}
\email{jack.ceroni@mail.utoronto.ca}
\author{Torin F. Stetina}
\affiliation{Simons Institute for the Theory of Computing, Berkeley, CA, 94704, USA}
\affiliation{Berkeley Quantum Information and Computation Center, University of California, Berkeley, CA 94720, USA}
\email{torins@berkeley.edu}
\author{M\'aria Kieferov\'a}
\affiliation{  Centre for Quantum Computation and Communication Technology, Centre for Quantum Software and Information, University of Technology Sydney, NSW 2007, Australia}
\email{maria.kieferova@uts.edu.au}
\author{Carlos Ortiz Marrero}
\affiliation{AI \& Data Analytics Division, Pacific Northwest National Laboratory,  Richland, WA 99354 }
\affiliation{Department of Electrical \& Computer Engineering, North Carolina State University,  Raleigh, NC 27607}
\email{carlos.ortizmarrero@pnnl.gov}
\author{Juan Miguel Arrazola}
\affiliation{Xanadu, Toronto, ON, M5G 2C8, Canada}
\email{juanmiguel@xanadu.ai}
\author{Nathan Wiebe}
\affiliation{ Department of Computer Science, University of Toronto, ON M5S 1A1, Canada}
\affiliation{High Performance Computing Group, Pacific Northwest National Laboratory, Richland, WA 99354}
\email{nawiebe@cs.toronto.edu}

\date{\today}

\begin{abstract}
The potential energy surface (PES) of molecules with respect to their nuclear positions is a primary tool in understanding chemical reactions from first principles. However, obtaining this information is complicated by the fact that sampling a large number of ground states over a high-dimensional PES can require a vast number of state preparations. In this work, we propose using a generative quantum machine learning model to prepare quantum states at arbitrary points on the PES. The model is trained using quantum data consisting of ground-state wavefunctions associated with different classical nuclear coordinates. Our approach uses a classical neural network to convert the nuclear coordinates of a molecule into quantum parameters of a variational quantum circuit. The model is trained using a fidelity loss function to optimize the neural network parameters. We show that gradient evaluation is efficient and numerically demonstrate our method's ability to prepare wavefunctions on the PES of hydrogen chains, water, and beryllium hydride. In all cases, we find that a small number of training points are needed to achieve very high overlap with the groundstates in practice. From a theoretical perspective, we further prove limitations on these protocols by showing that if we were able to learn across an avoided crossing using a small number of samples, then we would be able to violate Grover's lower bound. Additionally, we prove lower bounds on the amount of quantum data needed to learn a locally optimal neural network function using arguments from quantum Fisher information. This work further identifies that quantum chemistry can be an important use case for quantum machine learning.
\end{abstract}

\maketitle

\section{Introduction}
\label{sec:1}

One of the most widely-studied uses of quantum computers is the simulation and characterization of physical systems. Significant effort has been dedicated to developing quantum algorithms for a variety of calculations across many areas of physics, including condensed matter~\cite{abrams1997simulation, bauer2020quantum}, molecular physics~\cite{o2021efficient, delgado2021variational}, and quantum field theory~\cite{jordan2012quantum, shaw2020quantum}. However, in recent years, quantum chemistry has emerged as a leading application, and considerable focus has been placed on using quantum devices to determine the properties of molecules and materials~\cite{cao2019quantum, mcardle2020quantum, von2021quantum, lee2021even, sawaya2020resource, o2019calculating}. Most proposals for quantum computational chemistry algorithms attempt to solve the electronic structure problem, in which the nucleus of a molecule is assumed to be fixed, and the goal is to compute the ground-state energy of the electronic Hamiltonian.  While the general problem of computing the ground state of an arbitrary Hamiltonian is \QMA-complete~\cite{schuch2009computational}, under certain conditions, such as when given access to a sufficiently high-fidelity approximation of the true ground state ~\cite{kitaev2002classical, kempe2006complexity, gharibian2022improved}, quantum computers can efficiently yield accurate ground states. The most studied proposals for solving the electronic structure problem on circuit-based quantum computers are the variational quantum eigensolver (VQE)~\cite{peruzzo2014variational, mcclean2016theory, cerezo2021variational} for the noisy, intermediate scale regime, and quantum phase estimation (QPE)~\cite{kitaev1995quantum} for fault-tolerant quantum computers. Both techniques yield an approximation of the ground state energy of a Hamiltonian, as well as the ground state itself. 

While these approaches are effective in many scenarios, their use comes with a considerable cost. 
For example, for the FeMoCo molecule, a widely-used benchmark in quantum computational chemistry, the best estimate of QPE runtime for determining its ground state energy is just under 4 days~\cite{hypercontraction} on a fault-tolerant quantum computer with millions of physical qubits. On the other hand, VQE has the possibility of allowing for quantum computational chemistry with fewer, noisier qubits, but the number of measurements needed to estimate energies is often significant, making scaling of the algorithm to large molecules a challenge~\cite{gonthier2022measurements, wecker2015progress, huggins2021efficient, arrasmith2020operator}.

The challenges for both QPE and VQE worsen even further when they are considered in the context of solving practical problems in quantum chemistry. The electronic structure problem assumes a fixed molecular configuration, and therefore a fixed molecular Hamiltonian, of which we compute the ground state. In order to determine many dynamic or structural properties of molecules, such as reaction barriers and optimal geometries, a molecule must be studied in many different configurations. In general, characterizing this behaviour requires knowledge of a family of ground states for a set of Hamiltonians parameterized by classical nuclear coordinates $H(R)$. This is an arduous task, which requires computing many different ground states with corresponding energies lying on a high-dimensional potential energy surface. Running quantum algorithms such as QPE or VQE for each configuration independently would represent a significant computational cost even for molecules with a modest number of atoms.

We propose in this paper an alternative method for computing ground states corresponding to a wide range of molecular configurations, i.e., for reconstructing potential energy surfaces of molecules. A key motivation behind this work is that while the fixed nuclei electronic structure problem is already difficult, the ultimate goal of using quantum computers to compute accurate electronic energies requires sampling over many different nuclear configurations in a proposed chemical reaction coordinate. Therefore, the generation of many ground states, and subsequently energies and other properties, is of central interest in taking advantage of quantum computers for designing new materials and technologies. Instead of computing the ground states for many discrete molecular configurations independently, our algorithm uses a limited collection of data and, employing techniques from machine learning, builds a model that prepares the ground state over some region in parameter space.

Our work connects with recent progress for the task of learning from quantum mechanical data with both classical and quantum methods, an increasingly active area of research.  Notable results include demonstration that classical machine learning techniques are provably efficient for predicting and modelling certain properties of quantum many-body systems~\cite{huang2021provably}, and that quantum learning procedures can be more efficient than classical learning procedures for determining specific properties of certain unknown quantum states and processes~\cite{chen2021exponential, huang2021quantum, huang2021provably}. Our proposed algorithm shares some similarities with the above works, with the key difference being that the output of our model is a quantum state, rather than an estimate of some observable quantity or a classical approximation of a quantum state~\cite{aaronson2018shadow, huang2020predicting}. Therefore, algorithms of the form we proposed can be used to extract arbitrary ground state observables, or to output states that can be used in other quantum computational procedures requiring access to ground states. From the perspective of Ref.~\cite{schuld2018supervised}, this model falls into the quantum-quantum or ``QQ'' category of machine learning techniques: a quantum model trained with quantum data, and complements the growing literature on using quantum data and quantum machine learning to understand quantum systems (in particular, quantum chemical systems). Existing examples of QQ machine learning include learning excited states from ground state \cite{kawai2020predicting}, compression of quantum data \cite{romero2017quantum}, and learning of parametrized Hamiltonians ~\cite{PhysRevApplied.11.044087}, which has been applied to spin and molecular Hamiltonians under the name quantum meta-learning~\cite{PRXQuantum.2.020329}.

The algorithm proposed in this work is a hybrid classical-quantum generative model, in which we train a classical neural network to yield parameters, which when fed into a low-depth variational quantum circuit, approximate the corresponding ground state of $H(R)$ for a range of values of $R$. To train our model, we assume access to quantum data: ground states of $H(R)$ for a collection of coordinates $\{R_i\}_{i = 1}^{N}$, which can be loaded into a quantum computer. Since ground state preparation is a resource intensive task, the amount of quantum data needed to learn a model is a key metric for quantifying the efficiency and the feasibility of the algorithm. Ideally, a model of this form should generalize to new values of $R$ not contained in the training data. This allows us to generate approximations of previously unseen molecular ground states, at the more modest price of executing a shallow, variational quantum circuit for some set of parameters determined by a classical neural network.

To test these capabilities, we perform extensive numerical experiments for a collection of different molecules, and find that even with few data points, there is good generalization to unseen geometries in the potential energy surface. Ultimately, the aim of our proposal is to provide a concrete first step towards the development of practical techniques based on quantum machine learning for alleviating the cost of computing ground states of a parameterized molecular Hamiltonian. Under this framework, hard-to-obtain quantum data, originating from quantum algorithms or physical experiments~\cite{mcclean2021foundations}, is the resource that we attempt to leverage.

It is important to note that the particular generative model proposed is only one member of a large family of quantum-classical machine learning architecture for preparing ground states. Advanced models may utilize more sophisticated choices of cost function, classical neural network architecture, initialization, and circuit construction than those considered in this paper. The goal of this work is to both discuss the general concept of generative quantum machine learning applied to quantum chemistry, as well as provide a concrete example of what such a model would look like. As a result, we explore both practicalities associated with the particular model, including gradient sample complexity and numerical experiments, as well as more general considerations about abstract quantum state-learning procedures.

We begin in Sec.~\hyperref[sec:2]{II} by outlining a general architecture and training strategy for a generative model which prepares ground states of a parameterized Hamiltonian. In Sec.~\hyperref[sec:3]{III}, we discuss the details of training, and provide estimates on the sample complexity required for computing gradients of the model, which is necessary for optimization. In Sec.~\hyperref[sec:4]{IV}, to better understand the limits of quantum generative models, we introduce theoretical bounds related to data complexity of general quantum state-learning algorithms, and interpret these results in the context of our generative model. We conclude in Sec.~\hyperref[sec:5]{V} by providing numerics to support the quality of our model, demonstrating that one can effectively learn out-of-distribution ground states, and thus the resulting potential energy surfaces, of the H$_2$, H$_3^{+}$, H$_4$, BeH$_2$, and H$_2$O molecules, to a high degree of accuracy.

\section{A Generative Model for Preparing Electronic Ground States}
\label{sec:2}

Quantum chemistry has, in recent years, become a major focus for the development of quantum algorithms~\cite{berry2015hamiltonian,lee2021even,babbush2018low}.  The specific problem that has dominated most discussion surrounding quantum computing in this space is the electronic structure problem, which involves finding the minimum energy configuration for a molecule.  This problem reduces to the problem of computing the groundstate energy of a quantum system.  The electronic structure problem in a fixed basis is, in general,  \QMA-Hard~\cite{o2022intractability} 
Regardless, in many cases though a state with sufficient overlap can be identified and the overlap with the state can be computed within constant error $\epsilon$ using a polynomial number of quantum gates~\cite{aspuru2005simulated,whitfield2011simulation,wecker2015progress}.

The Hamiltonian for the electrons within the Born-Oppenheimer approximation (which states that the nuclear motion is uncoupled with the electronic motion) can be written as

\begin{equation}
    H(R) = \displaystyle\sum_{pq} h_{pq}(R) a_p^{\dagger} a_q + \frac{1}{2} \displaystyle\sum_{pqrs} h_{pqrs}(R) a_p^{\dagger} a_{q}^{\dagger} a_r a_s,
    \label{eqn:ham}
\end{equation}

where $a^{\dagger}_p$ and $a_p$ are the fermionic creation and annihilation operators acting on the $p$-th orbital, and

\begin{equation}
    h_{pq}(R) = \displaystyle\int dr \ \phi_p^{*}(r) \left( -\frac{\nabla^2}{2} - \displaystyle\sum_{I} \frac{Z_I}{|r - R_I|} \right) \phi_q(r),
    \label{eqn:one_elec}
\end{equation}

\begin{equation}
    h_{pqrs}(R) = \displaystyle\int dr_1 dr_2 \ \frac{\phi_p^{*}(r_1) \phi_q^{*}(r_2) \phi_r(r_2) \phi_s(r_1)}{|r_1 - r_2|},
    \label{eqn:two_elec}
\end{equation}

\noindent
are the one and two-electron integrals in the molecular orbital basis, $\phi_p(r)$, yielded from the Hartree-Fock optimization procedure. The one and two-electron integrals, and subsequently the molecular orbitals depend implicitly on $R$, which is the general nuclear coordinate associated with the fixed nuclear configuration in 3-dimensional space. We then utilize the Jordan-Wigner transform in order to map the fermionic creation and annihilation operators to qubit operators, which can be implemented in a quantum circuit \cite{wigner1928paulische, arrazola2021differentiable}.

The best known costs for estimating the groundstate energy at a specific nucelar configuration scales as~\cite{lee2021even} $\tilde{O}(N\lambda /\epsilon)$ where $N$ is the number of spin-orbitals used in the model and $\lambda$ is a quantity that depends on the details of the Hamiltonian and the basis set chosen.  For the case where a planewave basis is chosen an explicit scaling can be achieved of $\tilde{O}(N^3/\epsilon)$~\cite{babbush2018low}; however, such an explicit scaling is difficult to derive in a Gaussian basis (like those commonly used in chemistry) but can be numerically verified to scale polynomially in $N$ for most examples considered.

The challenge behind all of this is that the cost of preparing these groundstates is difficult.  It requires performing phase estimation to learn the groundstate energy within the eigenvalue gap, which leads to a cost of $O(N^3/|E_1 - E_0|)$ for the case of planewaves.  Such costs for generating a quantum state at any point on the potential energy surface can be challenging for problems with small gaps and so a question that naturally emerges is whether there exist faster methods for preparing approximate groundstates on quantum computers.  In general, we do not expect this to be possible without prior information; however, when approaching chemistry we almost never have a uniform prior.  We are frequently met with a host of prior experiments and often have systems that are smooth enough so nearby queries yield information about the quantum state in question.  Our aim is to use quantum machine learning as a paradigm to estimate new states on the potential energy surface of a molecule.  We discuss how generative quantum machine learning to approach this challenge below.

\subsection{Generative QML Models for Chemistry}

\begin{figure*}
\centering
    \includegraphics[width=1.5\columnwidth]{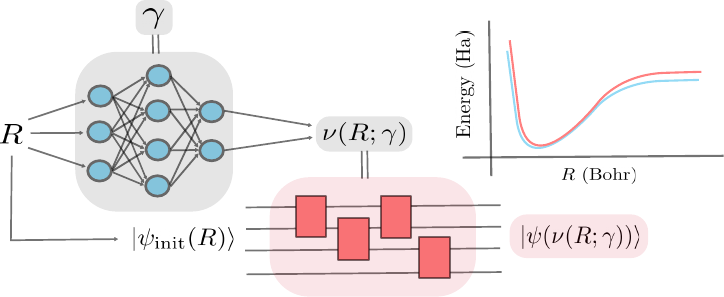}
    \caption{A diagram showing how a neural network is used to parameterize a map from coordinates $R$ to a state $|\psi(\nu(R; \gamma))\rangle$. The goal is to find the optimal parameters $\gamma = \gamma^{*}$ such that when they are used in the generative procedure, the resulting map from $R$ to $|\psi(\nu(R; \gamma^{*}))\rangle$ approximates the map from $R$ to the ground state $|\psi_0(R)\rangle$. This allows for accurate extraction of ground state observables, such as ground state energies, from our model, over $R$.}
    \label{fig:gqq}
\end{figure*}

The generative model that we propose contains both a classical and quantum component. We begin by fixing a parameterized quantum circuit $U(\theta)$. Although the structure of $U$ can vary, we assume that it can be expressed as a series of parameterized rotations~\cite{peruzzo2014variational},

\begin{equation}
    U(\theta) = \displaystyle\prod_{j = 1}^{N_P} e^{-i \theta_j H_j},
    \label{eq:ansatz}
\end{equation}
where each $H_j$ is Hermitian and chosen prior to learning. Given a parameterized Hamiltonian $H(R)$, the goal is to learn a map $R \mapsto \theta(R)$, such that $U(\theta(R))|\psi_{\text{init}}(R)\rangle$ approximates the ground state $|\psi_0(R)\rangle$ of $H(R)$ over some range of $R$, where $|\psi_{\text{init}}(R)\rangle$ is some pre-chosen initial state. For the sake of simplicity, $H(R)$ is taken to be non-degenerate over $R$ in the region considered. In order to find the mapping $R \mapsto \theta(R)$, we are given access to a collection of quantum data, $\mathcal{D}$, of the form

\begin{equation}
    \mathcal{D} = \left\{ \left( R_i, \left( \overbrace{|\psi_0(R_i)\rangle,\ \dots \ , \ |\psi_0(R_i)\rangle}^\text{$m_i$} \right) \right) \right\}_{i = 1}^{N},
\end{equation}
where each $R_i$ is unique and we have access to $m_i$ copies of a state $|\psi_0(R_i)\rangle$. We assume a polynomially large number of data, which should be contrasted with the exponentially large number of possible grid points $R_i$ that define a full potential energy surface in high dimensions. We do not explicitly consider the origin of the data set, but it can be created with several different methods, possibly by repeated application of QPE or VQE for the parameters $R_i$, or even by transduction of quantum states yielded from an experiment into a quantum device, as is envisioned in Ref.~\cite{mcclean2021foundations}. We note that the model is not restricted to returning parameterized ground states: given training data, it can be used to construct an approximation of any parameterized state. For example, one could imagine using a similar procedure to produce a model of a parameterized excited state $|\psi_j(R)\rangle$. We focus on ground state generation for the sake of concreteness, and due to the fact that computing ground states is a problem of practical importance.

Let $\nu(R; \gamma)$ denote a classical neural network, where $\gamma$ are the trainable parameters, which induces a function $R \mapsto \nu(R; \gamma)$. Using the data $\mathcal{D}$ along with $\nu$, we perform an optimization procedure which yields $\gamma^{*}$ such that the output state $U(\nu(R_i; \gamma^{*}))|\psi_{\text{init}}(R_i)\rangle$ approximates $|\psi_0(R_i)\rangle$, for each $R_i$ in the training data. Define $\theta(R) := \nu(R; \gamma^{*})$ to be the output of the model. The goal of such a training procedure is good generalization for all values of $R$. In other words,

\begin{align}
|\psi(\nu(R; \gamma^{*}))\rangle &:= U(\nu(R; \gamma^{*}))|\psi_{\text{init}}(R)\rangle \nonumber\\ & = U(\theta(R))|\psi_{\text{init}}(R)\rangle,
\end{align}
should approximate $|\psi_0(R)\rangle$ for a wider range of $R$ outside of the training set. To find the optimal parameters $\gamma^{*}$, we minimize a cost function $C(\gamma)$ over the training data. In this work, we choose $C$ to be the average infidelity of the state produced by the model with each unique training example (each state in $\mathcal{D}$ corresponds to a unique $R_i$). More specifically, the cost function is

\begin{equation}
C(\gamma) = 1 - \frac{1}{N} \displaystyle\sum_{i = 1}^{N} |\langle \psi_0(R_i) | \psi( \nu(R_i; \gamma) ) \rangle|^2,
\label{eq:cost}
\end{equation}

\noindent
where for sufficiently expressive $U$, minimization of $C$ will result in $\gamma^{*}$ such that $|\psi(\nu(R_i; \gamma^{*}))\rangle \approx |\psi_0(R_i)\rangle$ for each $R_i$.  One immediate benefit of using the infidelity cost function is that it admits relatively simple gradients of the form

\begin{multline}
    \frac{\partial C(\gamma)}{\partial \gamma_k} = - \frac{1}{N} \displaystyle\sum_{i = 1}^{N} \Bigg( 2 \text{Re} \left[ \langle \psi_0(R_i) | \psi(\nu(R_i; \gamma)) \rangle \right] \\ \displaystyle\sum_{a = 1}^{N_P} \frac{\partial \nu(R_i; \gamma)_a}{\partial \gamma_k} \Big\langle \psi_0(R_i) \Big| \frac{\partial \psi(\theta)}{\partial \theta_a} \Big\rangle \biggr\rvert_{\theta = \nu(R_i; \gamma)} \Bigg),
    \label{eq:grad}
\end{multline}
where $|\psi(\theta)\rangle := U(\theta) |\psi_{\text{init}}(R_a)\rangle$ for each term in the sum and $N_P$ is the output dimension of the classical neural network, i.e, the number of parameters of the quantum circuit (for a derivation, see Sec.~\ref{sec:3}). Overlaps between the training states, and the output/output derivatives of the parameterized circuit $U$ can be computed via sampling, or amplitude estimation and the linear combination of unitaries (LCU) method~\cite{childs2012hamiltonian, berry2015hamiltonian}, while the derivatives of $\nu$ can be computed efficiently in the case of deep neural networks via backpropagation \cite{rumelhart1986learning}. With a procedure for computing derivatives, it is then possible to minimize $C$ via gradient-based optimization methods, such as gradient descent or Adam~\cite{kingma2014adam}. The computational cost of computing gradients of the cost function in Eq.~\eqref{eq:cost} is discussed in detail in Section~\hyperref[sec:3]{III}.

We conclude this section by discussing one of the central challenges in designing quantum machine learning models: barren plateaus, and how the proposed model may be robust to this issue. The barren plateaus problem is the observation that many classes of parameterized quantum circuits suffer from exponentially vanishing gradients over large parts of the parameter space~\cite{mcclean2018barren}. Barren plateaus highlight the necessity to incorporate inductive biases into parameterized quantum circuits based on the structure of a particular problem, by choosing appropriate circuit ansatzes, cost function, and initialization, among other hyperparameters. As is discussed in Refs.~\cite{mcclean2018barren, cerezo2021cost}, $C(\gamma)$ of the general form in Eq.~\eqref{eq:grad} is conducive to barren plateaus. This issue can, in some cases, be resolved by choosing a more sophisticated cost function \cite{renyi_divergence_training, cerezo2021cost}. However, before swapping $C$ for one of these functions, it is important to note that any generative model tasked with returning a molecular ground state can have as its initial guess an approximation of the true ground state. In our work we set $|\psi_\text{init}(R)\rangle := |\psi_{\text{HF}}(R)\rangle$, where $|\psi_{\text{HF}}(R)\rangle$ is the Hartree-Fock state at geometry $R$; an approximate solution to the electronic Schrodinger equation which can be computed efficiently on a classical device.

The parameterized quantum circuit $U$ can then be thought of as applying corrections to the Hartree-Fock state. This choice makes the initialization of our model far from random. In addition, since we have access to the Hamiltonian $H(R)$ of which we are attempting to prepare the ground state, we can use this knowledge to tailor our circuit ansatz to each particular molecule considered, using the ADAPT-VQE algorithm~\cite{grimsley2019adaptive} (see Sec.~\ref{sec:5} for details). Ref.~\cite{grimsley2022adapt} provides empirical evidence that in many cases, adaptive circuits do not suffer from the issue of barren plateaus. It is therefore reasonable to hypothesize that initialization in the Hartree-Fock state and an adaptively-prepared circuit will, in many cases, constrain the model's optimization to a region that does not suffer from the barren plateau problem. This conclusion is supported in the numerics (Sec.~\ref{sec:5}), where we observe good performance when using the infidelity cost function of Eq.~\eqref{eq:cost} for training.

\section{Gradient Sampling Complexity}
\label{sec:3}

While the model described in Section~\hyperref[sec:2]{II} is a hybrid quantum-classical procedure, the part of the algorithm that is executed on a quantum device can be reduced to the calculation of gradients of the cost function $C(\gamma)$. As a result, understanding the quantum sample complexity of gradient calculations allows us to better understand the cost of running the entire algorithm. The other important aspect is the number of training steps required to find the optimal $\gamma^{*}$, but this task is much more challenging, and likely varies considerably on a case-by-case basis.

In this section, we discuss the sample complexity required to compute gradients of the average infidelity cost function introduced in Section~\hyperref[sec:2]{II}. Recall that $|\psi(\theta)\rangle := U(\theta)|\psi_{\text{init}}(R_i)\rangle$ for some $R_i$. We can express the gradient of the $i$-th overlap terms in Eq.~\eqref{eq:cost} as

\begin{align}
    &\frac{\partial}{\partial \gamma_k} | \langle \psi_0(R_i) | \psi(\nu(R_i; \gamma)) \rangle |^2 \nonumber \\
    &= \frac{\partial}{\partial \gamma_k} \langle \psi_0(R_i) | \psi(\nu(R_i; \gamma)) \rangle \langle \psi(\nu(R_i; \gamma)) |\psi_0(R_i)\rangle \nonumber\\
    &= \braket{\psi(\nu(R_i; \gamma))}{\psi_0(R_i)} \Big\langle \psi_0(R_i) \Big| \frac{\partial \psi(\nu(R_i; \gamma))}{\partial \gamma_k} \Big\rangle + \text{h.c.} \nonumber\\
    &= 2 \text{Re} \left[ \braket{\psi(\nu(R_i; \gamma))}{\psi_0(R_i)} \Big\langle \psi_0(R_i) \Big| \frac{\partial \psi(\nu(R_i; \gamma))}{\partial \gamma_k} \Big\rangle \right] \nonumber\\ 
    &2 \text{Re} \Bigg[ \braket{\psi(\nu(R_i; \gamma))}{\psi_0(R_i)} \displaystyle\sum_{a = 1}^{N_P} \frac{\partial \nu(R_i; \gamma)_a}{\partial \gamma_k} \nonumber \\ & \times \Big\langle \psi_0(R_i) \Big| \frac{\partial \psi(\theta)}{\partial \theta_a} \Big\rangle \biggr\rvert_{\theta = \nu(R_i; \gamma)} \Bigg].
    \label{eq:gr}
\end{align}

Now, note that

\begin{align}
    \ket{\frac{\partial \psi(\theta)}{\partial \theta_a}} &= \frac{\partial}{\partial \theta_a} U(\theta) |\psi_{\text{init}}(R_i)\rangle \nonumber\\ 
    &= \frac{\partial}{\partial \theta_a} \Bigg[ \displaystyle\prod_{j = 1}^{N_P} e^{-i \theta_j H_j} \Bigg] |\psi_{\text{init}}(R_i)\rangle \nonumber\\
    &= -i \Bigg[ \displaystyle\prod_{j < a} e^{-i \theta_j H_j} H_a \displaystyle\prod_{j \geq a} e^{-i \theta_j H_j} \Bigg] |\psi_{\text{init}}(R_i)\rangle \nonumber\\
    &= -i \Bigg[ \displaystyle\prod_{j < a} e^{-i \theta_j H_j} H_a \Bigg( \displaystyle\prod_{j < a} e^{-i \theta_j H_j} \Bigg)^{\dagger} \nonumber \\ & \times \displaystyle\prod_{j < a} e^{-i \theta_j H_j} \displaystyle\prod_{j \geq a} e^{-i \theta_j H_j} \Bigg] |\psi_{\text{init}}(R_i)\rangle \nonumber\\
    &= -i \Bigg[ \displaystyle\prod_{j < a} e^{-i \theta_j H_j} H_a \Bigg( \displaystyle\prod_{j < a} e^{-i \theta_j H_j} \Bigg)^{\dagger} \Bigg] |\psi(\theta)\rangle \nonumber\\ 
    &= -i \hat{H}_a(\theta) |\psi(\theta)\rangle,
    \label{eq:sg}
\end{align}

\noindent
where $\hat{H}_a(\theta) := \prod_{j < a} e^{-i \theta_j H_j} H_a \left( \prod_{j < a} e^{-i \theta_j H_j} \right)^{\dagger}$. Let $\hat{H}_a := \hat{H}_a(\nu(R_i; \gamma))$. For brevity, we will henceforth refer to $|\psi(\nu(R_i; \gamma))\rangle$ as $|\psi\rangle$, and $|\psi_0(R_i)\rangle$ as $|\psi_i\rangle$. Using Eq.~\eqref{eq:gr} and Eq.~\eqref{eq:sg}, we get

\begin{align}
    & \frac{\partial}{\partial \gamma_k} | \langle \psi_0(R_i) | \psi(\nu(R_i; \gamma)) \rangle |^2 = \frac{\partial}{\partial \gamma_k} | \langle \psi_i | \psi \rangle |^2 \nonumber \\
    &= 2 \text{Re} \left[ -i \braket{ \psi }{\psi_i} \displaystyle\sum_{a = 1}^{N_P} \frac{\partial \nu(R_i; \gamma)_a}{\partial \gamma_k} \langle \psi_i | \hat{H}_a \ket{\psi} \right] \nonumber \\ 
    &= 2 \displaystyle\sum_{a = 1}^{N_P} \frac{\partial \nu(R_i; \gamma)_a}{\partial \gamma_k} \text{Im} \left[ \bra{\psi_i} \hat{H}_a \ket{\psi} \braket{ \psi}{\psi_i} \right] \nonumber \\
    &= 2 \displaystyle\sum_{a = 1}^{N_P} \frac{\partial \nu(R_i; \gamma)_a}{\partial \gamma_k} \text{Im} \left[ \bra{\psi_i} \hat{H}_a \frac{\mathbbm{1} - (\mathbbm{1} - 2 \ket{\psi} \bra{ \psi})}{2} \ket{\psi_i} \right] \nonumber \\ 
    &= \displaystyle\sum_{a = 1}^{N_P} \frac{\partial \nu(R_i; \gamma)_a}{\partial \gamma_k} \Big[ \text{Im} \left[ \bra{\psi_i} \hat{H}_a \ket{\psi_i} \right] \nonumber 
    \\ & + \text{Im} \Big[ \bra{\psi_i} \hat{H}_a (2 \ket{ \psi} \bra{ \psi} - \mathbbm{1} ) \ket{\psi_i} \Big] \Big],
    \label{eq:final_grad}
\end{align}

\noindent
where $2 | \psi \rangle \langle \psi | - \mathbbm{1}$ is unitary and is given by
$2 | \psi\rangle \langle \psi| - \mathbbm{1} = U(\nu(R_i; \gamma)) V_{\text{init}}(R_i) (2 |0 \rangle \langle 0 | - \mathbbm{1}) V^{\dagger}_{\text{init}}(R_i) U^{\dagger}(\nu(R_i; \gamma))$, where $V_{\text{init}}(R_i) |0\rangle = |\psi_{\text{init}}(R_i)\rangle$. 

\subsection{Computing Gradients with the Hadamard Test}
\label{sec:had_test}

The first method we propose for computing gradients of the generative model relies on performing the Hadamard test for each of the individual terms in the sum of Eq.~\ref{eq:final_grad}. In particular, we arrive at the following result on the total number of states required to compute a derivative of $C(\gamma)$. 

\begin{theorem}[Hadamard Test Gradient Sample Complexity]
The number of queries to the state preparation unitary $U(\nu(R_i;\gamma))$ and the controlled $\hat{H}_a$ operation, $M$, required to compute a partial derivative of the infidelity cost function $C(\gamma)$ defined in Eq.~\eqref{eq:cost} with respect to parameter $\gamma_k$, to within error $\epsilon$ and probability at least $1 - \delta$, scales as

\begin{equation}
    M \in O\left(\frac{N}{\epsilon^2} \log \left( \frac{1}{\delta} \right) \left( \max_i\sum_{j = 1}^{N_P} \left| \frac{\partial \nu(R_i; \gamma)_j}{\partial \gamma_k} \right| \right)^2 \right),
\end{equation}
using a Hadamard test-based approach from Fig.~\ref{fig:had_test} to gradient estimation. $N$ is the number of unique data states, and $\nu$ is the classical neural network.
\end{theorem}

\begin{proof}
Since each Hermitian $\hat{H}_a$ can be expanded as a linear combination of unitaries $\hat{H}_a = \sum_{l} c_a^{l} V_a^{l}$, we can compute an estimate of the above sum by performing the Hadamard test to estimate the values of

\begin{equation}
\Im\Big[\bra*{\psi_i} V_a^{l} \ket{\psi_i}\Big] \ \ \ \ \text{and} \ \ \ \ \Im\Big[\bra*{\psi_i} V_a^{l} (2\ket{\psi}\bra{\psi}-\mathbbm{1}) \ket{\psi_i}\Big],
\label{eq:q}
\end{equation}
for each $V_a^{l}$. Note that there are many cases where the $\hat{H}_a$ are themselves unitary. Clearly, these will be precisely the cases where the generators $H_a$ are unitary, which occurs often in Trotterized circuit ansatze (for example, when each $H_a$ is a tensor product of Pauli operations). For simplicity, assume that each of the $\hat{H}_a$ are unitary, so we can simply compute expectation values without decomposing into a sum of multiple expectation values. See Fig.~\ref{fig:had_test} for circuit diagrams of the Hadamard test corresponding to both quantities.

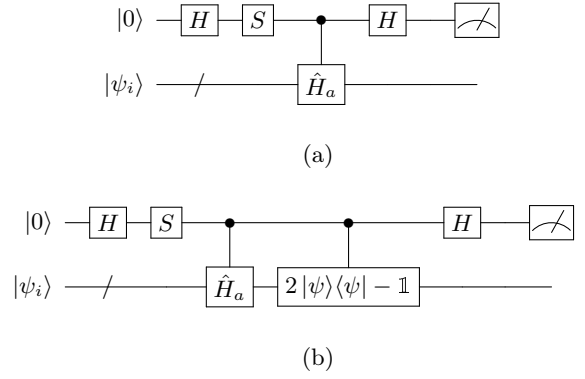
\begin{figure}
\centering
\begin{subfigure}[b]{0.5\columnwidth}
         \centering
\[
\Qcircuit @C=1em @R=1em @!R {
\lstick{\ket{0}} & \gate{H} & \gate{S} & \ctrl{1} & \gate{H} & \qw & \meter\\
\lstick{\ket{\psi_i}} & {/} \qw & \qw & \gate{\hat{H}_a}  & \qw & \qw & \qw
}
\]
\caption{}
\end{subfigure}
\begin{subfigure}[b]{0.4\textwidth}
\centering
\[
\Qcircuit @C=1em @R=1em @!R {
\lstick{\ket{0}} & \gate{H} & \gate{S} & \ctrl{1} & \ctrl{1}&\gate{H} & \qw & \meter\\
\lstick{\ket{\psi_i}} & {/} \qw & \qw & \gate{\hat{H}_a} &\gate{2\ket{\psi}\!\bra{\psi}-\mathbbm{1}}  & \qw & \qw & \qw
}
\]
\caption{}
\end{subfigure}

\caption{Hadamard test circuits for computing (a) $\Im\Big[\bra*{\psi_i} \hat{H}_a  \ket{\psi_i}\Big]$ and  $\Im\Big[\bra*{\psi_i} \hat{H}_a (2\ket{\psi}\bra{\psi}-\mathbbm{1})\ket{\psi_i}\Big]$ in (b) . Note that $S=\sqrt{Z}$ represents a phase gate. Here  $\hat{H}_a(\theta) := \prod_{j < a} e^{-i \theta_j H_j} H_a \left( \prod_{j < a} e^{-i \theta_j H_j} \right)^{\dagger}$.}
\label{fig:had_test}
\end{figure}

Now that we have outlined the calculations required to compute the gradients of each term in the cost function, we turn our attention to the question of how many samples are required in order to compute accurate estimates of the gradients, with high probability.

For each $\hat{H}_a$, the Hadamard test circuits depicted in Fig.~\ref{fig:had_test} yield a random variable with expectation values equal to the quantities in Eq.~\eqref{eq:q}, when measured in the $Z$-basis. Define the random variables $X_a^{i}$ with values in $\{-1, 1\}$ for $a \in \{1, \ ..., \ 2N_P\}$ to be the random variables corresponding to each of the unique $2N_P$ Hadamard test circuits, multiplied by the corresponding sign of the neural network gradient (so we can later take positive linear combinations of these random variables). These random variables therefore have expectation values

\begin{equation}
    \mathbb{E}(X_a^{i}) = \text{sign} \left( \frac{\partial \nu(R_i; \gamma)_a}{\partial \gamma_k} \right) \text{Im} \left[ \langle \psi_i | \hat{H}_a | \psi_i \rangle \right]
\end{equation}

\noindent
for $a \in \{1, \ ..., \ N_P\}$, and

\begin{equation}
    \mathbb{E}(X_a^{i}) = \text{sign} \left( \frac{\partial \nu(R_i; \gamma)_{a - N_P}}{\partial \gamma_k} \right) \text{Im} \left[ \langle \psi_i | \hat{H}_a (2 |\psi\rangle \langle\psi| - \mathbbm{1}) |\psi_i\rangle \right]
\end{equation}

\noindent
for $a \in \{N_P + 1, \ ..., \ 2N_P\}$. Therefore, for $a \in \{1, ..., \ N_P\}$, the random variable $X_a^{i}$ is realized by sampling from the circuit depicted in Fig.~\ref{fig:had_test}(a) and for $a \in \{N_P + 1, \ ..., \ 2N_P\}$, the random variable $X_a^{i}$ comes from the circuit in Fig.\ref{fig:had_test}(b). Since we are interested in computing a sum of such random variables, we introduce new random variables

\begin{equation}
    S^{i} = \displaystyle\sum_{a = 1}^{2N_P} p_a^{i} X_a^{i},  \ \ \ \ \text{where} \ \ \ \
    p_a^{i} = \frac{ \left| \frac{\partial \nu(R_i; \gamma)_a}{\partial \gamma_k} \right| }{2 \sum_{j = 1}^{N_P} \left| \frac{\partial \nu(R_i; \gamma)_j}{\partial \gamma_k} \right| }.
\end{equation}
It then follows from Eq.~\eqref{eq:final_grad} that the gradient term we wish to estimate is proportional to $\mathbb{E}(S^{i})$,

\begin{equation}
\frac{\partial}{\partial \gamma_k} | \langle \psi_i | \psi \rangle |^2 = 2 \left( \displaystyle\sum_{j = 1}^{N_P} \left| \frac{\partial \nu(R_i; \gamma)_j}{\partial \gamma_k} \right| \right) \mathbb{E}(S^{i}).
\end{equation}

We can sample from $S^{i}$ by choosing $a$ in accordance with the probability distribution defined by the $p_a$, and then take a sample from the corresponding $X_a^{i}$. Note that $-1 \leq S^{i} \leq 1$. We can use Hoeffding's inequality to conclude that after taking $m_i$ samples from $S^{i}$, the probability that the estimated average $S^{i}_{m_i}$ will differ from $\mathbb{E}(S^{i})$ by a magnitude greater than or equal to $\epsilon$ is upper-bounded by

\begin{equation}
    \text{Pr} \big( \left| S_{m_i}^{i} - \mathbb{E}(S^{i}) \right| \geq \epsilon \big) \leq 2 \exp \left( - \frac{m_i \epsilon^2} {2} \right).
    \label{eq:hoeffding}
\end{equation}
It follows that for the estimator of the gradient, $G^{i}_{m_i} = 2 \left( \sum_{j = 1}^{N_P} \left| \frac{\partial \nu(R_i; \gamma)_j}{\partial \gamma_k} \right| \right) S_{m_i}^{i}$, we have, 

\begin{align}
    &\text{Pr} \left( \left| G_{m_i}^{i} - \frac{\partial}{\partial \gamma_k} | \langle \psi_i | \psi \rangle |^2 \right| \geq \epsilon \right) \nonumber \\ 
    &= \text{Pr} \left( 2 \left( \sum_{j = 1}^{N_P} \left| \frac{\partial \nu(R_i; \gamma)_j}{\partial \gamma_k} \right| \right) \left| S_{m_i}^{i} - \mathbb{E}(S^{i}) \right| \geq \epsilon \right) \nonumber \\ 
    &= \text{Pr} \left( |S_{m_i}^{i} - \mathbb{E}(S^{i}) | \geq \frac{\epsilon}{2 \left( \sum_{j = 1}^{N_P} \left| \frac{\partial \nu(R_i; \gamma)_j}{\partial \gamma_k} \right| \right)} \right) \nonumber \\
    &\leq 2 \exp \left(- \frac{m_i \epsilon^2}{8 \left( \sum_{j = 1}^{N_P} \left| \frac{\partial \nu(R_i; \gamma)_j}{\partial \gamma_k} \right| \right)^2} \right).
\end{align}
Therefore, we can choose $m_i$ such that

\begin{equation}
    2 \exp \left(- \frac{m_i \epsilon^2}{8 \left( \sum_{j = 1}^{N_P} \left| \frac{\partial \nu(R_i; \gamma)_j}{\partial \gamma_k} \right| \right)^2} \right) \leq \delta,
\end{equation}
which ensures that we have an $\epsilon$-good estimate of the gradient with probability greater than $1 - \delta$. Solving for a number of samples $m_i$ that are sufficient for an $\epsilon$-approximation gives

\begin{equation}
   m_i \geq \frac{8}{\epsilon^2} \log \left( \frac{2}{\delta} \right) \left( \sum_{j = 1}^{N_P} \left| \frac{\partial \nu(R_i; \gamma)_j}{\partial \gamma_k} \right| \right)^2
   \label{eq:bound}.
\end{equation}

We now have a lower bound on the number of samples required to compute the gradient with respect to $\gamma_k$ of the $i$-th term of the sum which gives $C(\gamma)$. Recall that 

\begin{equation}
    \frac{\partial C(\gamma)}{\partial \gamma_k} = - \frac{1}{N} \displaystyle\sum_{i = 1}^{N} \frac{\partial}{\partial \gamma_k} | \langle \psi_i | \psi \rangle |^2,
\end{equation}
so it follows that if we define the estimator $G_m = -\frac{1}{N} \sum_{i = 1}^{N} G^{i}_{m_i}$ of $\frac{\partial C(\gamma)}{\partial \gamma_k}$, then 

\begin{align}
    &\text{Pr} \left( \left| G_m - \frac{\partial C(\gamma)}{\partial \gamma_k} \right| \geq \epsilon \right) \nonumber \\ 
    &= \text{Pr} \left( \frac{1}{N} \left| \displaystyle\sum_{i = 1}^{N} \left( G_{m_i}^{i} - \frac{\partial}{\partial \gamma_k} | \langle\psi_i | \psi \rangle|^2 \right) \right| \geq \epsilon \right) \nonumber\\
    &\leq \text{Pr} \left( \frac{1}{N} \displaystyle\sum_{i = 1}^{N} \Big| G_{m_i}^{i} - \frac{\partial}{\partial \gamma_k} | \langle \psi_i | \psi \rangle|^2\Big| \geq \epsilon \right) \nonumber \\
    &\leq \text{Pr} \left( \max_{i} \Big| G_{m_i}^{i} - \frac{\partial}{\partial \gamma_k} |\langle \psi_i | \psi \rangle|^2 \Big| \geq \epsilon \right).
    \label{eq:new_eq}
\end{align}

Thus, if we choose each $m_i$ according to Eq.~\eqref{eq:bound}, it guarantees that $\text{Pr} \left( | G_{m_i}^{i} - \frac{\partial}{\partial \gamma_k} |\langle \psi_i | \psi \rangle|^2 | \geq \epsilon \right) \leq \delta$ for all $i$. This implies that the inequality in Eq.~\eqref{eq:new_eq} will also be bounded above by $\delta$ and hence we will be able to estimate the gradient of our model with respect to $\gamma_k$, with precision $\epsilon$, with probability greater than or equal to $1 - \delta$. Choosing $m_i$ to saturate its corresponding lower limit in Eq.~\ref{eq:bound} gives a sufficient number of samples for each term. Summing the result gives that the total number of samples needed scales as
\begin{equation}
    M= \sum_{i=1}^N m_i \in O\left(\frac{N}{\epsilon^2} \log \left( \frac{1}{\delta} \right) \left( \max_i\sum_{j = 1}^{N_P} \left| \frac{\partial \nu(R_i; \gamma)_j}{\partial \gamma_k} \right| \right)^2 \right),
\end{equation}
where we have used the fact that $\sum_{i = 1}^{N} \sum_{j = 1}^{N_P} \left| \frac{\partial \nu(R_i; \gamma)_j}{\partial \gamma_k} \right| \leq N \max_i\sum_{j = 1}^{N_P} \left| \frac{\partial \nu(R_i; \gamma)_j}{\partial \gamma_k} \right|$.

\end{proof}

\subsection{A Coherent Approach to Gradient Estimation}

The approach outlined in the previous section uses short-depth circuits for gradient estimation, but it is also possible to take another approach to gradient estimation where we assume it is possible to execute circuits of greater depth, which yields better complexity with respect to $\epsilon$. In particular, we can utilize the linear combination of unitaries (LCU) method in order to estimate the expectation value in Eq.~\eqref{eq:final_grad} without breaking it up into individual terms, which reduces the scaling in circuit executions from $O(1/\epsilon^2)$ to $O(1/\epsilon)$. 

\begin{theorem}[Coherent Gradient Sample Complexity]
The number of queries to the ansatz state preparation operator $U(R;\gamma)$ and the controlled $\hat{H}_a$ operation, $M$, needed to use the circuit in Fig.~\ref{fig:coh_test}  to compute a partial derivative of the infidelity cost function $C(\gamma)$ defined in Eq.~\ref{eq:cost} with respect to parameter $\gamma_k$, to within error $\epsilon$ and probability $1 - \delta$, scales as

\begin{equation}
    M \in  O \left( \frac{N}{\epsilon} \log \left( \frac{1}{\delta} \right) \max_i \displaystyle\sum_{j = 1}^{N_P} \left| \frac{\partial \nu(R_i;\gamma)_j}{\partial \gamma_k} \right| \right),
\end{equation}
when using a coherent approach to gradient estimation. $N$ is the number of unique data states and $\nu$ is the classical neural network function.
\end{theorem}

\begin{proof}
Again, assume that each $\hat{H}_a$ is unitary, for simplicity. Note that from Eq.~\eqref{eq:final_grad}, we can write the gradient of the $i$-th term of $C(\gamma)$ as an expectation value of a sum of $2N_P$ unitaries,

\begin{align}
    \label{eq:n_grad}
    &\frac{\partial}{\partial \gamma_k} |\langle \psi_i | \psi \rangle |^2 = \text{Im} \Bigg[ \Big\langle \psi_i \Big| \displaystyle\sum_{a = 1}^{N_P} \frac{\partial \nu(R_i; \gamma)_a}{\partial \gamma_k} \\ &(\hat{H}_a + \hat{H}_a (2 |\psi \rangle \langle \psi | - \mathbbm{1}) \Big|\psi_i\Big\rangle \Bigg] = \text{Im} \left[ \Big\langle \psi_i \Big| \displaystyle\sum_{a = 1}^{2N_P} \beta_a V_a \Big| \psi_i \Big\rangle \right] \nonumber
\end{align}

where

\begin{equation}
    \beta_a = \left| \frac{\partial \nu(R_i; \gamma)_a}{\partial \gamma_k} \right| \ \ \ \text{and} \ \ \ V_a = \text{sign} \left( \frac{\partial \nu(R_i; \gamma)_a}{\partial \gamma_k} \right) \hat{H}_a
\end{equation}
for $a \in \{1, \ ..., \ N_P\}$, and

\begin{align}
    \beta_a &= \left| \frac{\partial \nu(R_i; \gamma)_{a - N_P}}{\partial \gamma_k} \right| \nonumber \ \ \ \text{and} \\ V_a &= \text{sign} \left( \frac{\partial \nu(R_i; \gamma)_{a - N_P}}{\partial \gamma_k} \right) \hat{H}_a (2 |\psi\rangle \langle \psi| - \mathbbm{1}).
\end{align}
for $a \in \{N_P + 1, \ ..., \ 2N_P\}$. Now, as is standard in any LCU procedure, we define the PREP and SELECT operators. PREP will act on a qubit ancilla of size $s = \lceil \log_2 2N_P \rceil$ as

\begin{equation}
    \text{PREP} |0\rangle^{\otimes s} = \frac{1}{\sqrt{\beta}} \displaystyle\sum_{a = 1}^{2N_P} \sqrt{\beta_a} |a\rangle \ \ \ \text{where} \ \ \ \beta = \displaystyle\sum_{a = 1}^{2N_P} \beta_a.
\end{equation}

Define SELECT as a controlled operation acting on the ancilla as well as a qubit register encoding $|\psi_i\rangle$, such that

\begin{equation}
    \text{SELECT} |a \rangle |\psi_i\rangle = |a\rangle V_a |\psi_i\rangle
\end{equation}
for $a \in \{1, \ ..., \ 2N_P\}$. In order to implement a procedure which returns the desired expectation value in Eq.~\eqref{eq:n_grad}, we utilize an LCU method which begins with a qubit register encoding the state $|\psi_i\rangle$ along with an ancilla register $|0\rangle |0\rangle^{\otimes s}$ so that the initial state of the circuit is $|0\rangle |0\rangle^{\otimes s} |\psi_i\rangle$. We then apply a sequence of gates, including PREP and SELECT, as defined above. See Fig.~\ref{fig:coh_test} for a circuit diagram of the gate sequence applied to the initial state. The circuit in Fig.~\ref{fig:coh_test} results in the state

\begin{figure}
\centering
\[
\Qcircuit @C=1em @R=1em @!R {
\lstick{\ket{0}} & \gate{H} & \gate{S} & \ctrl{2} & \gate{H} & \qw\\
\lstick{\ket{0}} & {/} \qw & \gate{\text{PREP}} & \ctrl{1} & \qw & \qw \\
\lstick{\ket{0}} & {/} \qw & \gate{U_i} & \gate{\text{SELECT}}  & \qw & \qw
}
\]
\caption{An LCU circuit for computing $\text{Im} \left[ \langle \psi_i | \sum_{a = 1}^{2N_P} \beta_a V_a | \psi_i \rangle \right]$. This circuit is repeated multiple times during amplitude estimation where a state is considered marked if the top-most qubit is $\ket{0}$. Note that $U_i | 0\rangle = |\psi_i\rangle$.} 
\label{fig:coh_test}
\end{figure}
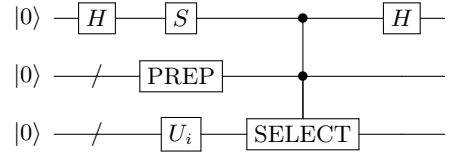

\begin{equation}
    |\psi_{\text{out}}\rangle = \frac{1}{2} \displaystyle\sum_{a = 1}^{2N_P} \sqrt{\frac{\beta_a}{\beta}} \left[ |0\rangle |a\rangle (\mathbbm{1} + i V_a) |\psi_i\rangle + |1\rangle|a\rangle (\mathbbm{1} - i V_a) |\psi_i\rangle \right].
\end{equation}
It follows that the probability of measuring the first qubit in state $|0\rangle$ is given by

\begin{multline}
\frac{1}{4} \left( \displaystyle\sum_{a = 1}^{2N_p} \sqrt{\frac{\beta_a}{\beta}} \langle a| \langle \psi_i | (\mathbbm{I} - iV_a^{\dagger}) \right) \left( \displaystyle\sum_{a = 1}^{2N_p} \sqrt{\frac{\beta_a}{\beta}} |a\rangle (\mathbbm{I} + iV_a) |\psi_i\rangle \right) \\ = \frac{1}{2} - \frac{1}{2} \ \text{Im} \left[ \langle \psi_i | \left( \sum_{a = 1}^{2N_P} \frac{\beta_a}{\beta} V_a \right) | \psi_i \rangle \right],
\end{multline}
where we use the fact that 

\begin{multline}
\langle \psi_i | (\mathbbm{1} - iV_a^{\dagger})(\mathbbm{1} + iV_a) |\psi_i\rangle = 2 + i \left[ \langle \psi_i | V_a | \psi_i \rangle - \langle \psi_i | V_a^{\dagger} |\psi_i\rangle \right] \\ = 2 - 2 \ \text{Im} \langle \psi_i | V_a |\psi_i\rangle.
\end{multline}
This allows us to compute the quantity in Eq.~\eqref{eq:n_grad}.

Since we are operating in the regime of fault-tolerant quantum circuits, we make use of quantum amplitude estimation to derive better bounds of the number of samples required to estimate the above probability. More specifically, we utilize the following theorem \cite{brassard2002quantum}.

\begin{theorem}[Amplitude Estimation]
For integer $R$, integer $k \geq 2$, unitary $V$, and projector $P$ such that $| \langle 0 | V^{\dagger} P V | 0\rangle |^2 = p$, there exists an algorithm that outputs an estimate $\widetilde{p}$ of $p$ such that

\begin{align}
    |\widetilde{p} - p| &\leq 2\pi k \frac{\sqrt{p ( 1 - p)}}{R} + k^2 \frac{\pi^2}{R^2} \nonumber \\ &\leq \frac{k\pi}{R} + \frac{k^2 \pi^2}{R^2} \leq \frac{2k^2 \pi^2}{R}
\end{align}

with probability greater than $1 - \frac{1}{2(k - 1)}$, where $O(R)$ evaluations of the unitaries $V$ and $\mathbbm{1} - 2P$ are used.
\label{thm:ae}
\end{theorem}

We set $V$ equal to the circuit in Fig.~\ref{fig:coh_test}, and let $P$ be the projector $|0\rangle \langle 0 |$ onto the first qubit. Note that for this procedure to work, we must apply unitaries $U_i$ which prepare our data states, $U_i | 0\rangle = |\psi_i\rangle$ multiple times in series. This eliminates the possibility of using this method of gradient calculation in the case that the data states $|\psi_i\rangle$ are obtained from experiments and fed into a quantum computer, without the ability to implement their corresponding state-preparation unitary on a quantum device. We have

\begin{align}
    p = | \langle 0 | V^{\dagger} P V | 0 \rangle|^2 &= | \langle \psi_{\text{out}} | P | \psi_{\text{out}} \rangle |^2 \nonumber \\ & = \frac{1}{2} - \frac{1}{2} \text{Im} \left[ \bra{\psi_i} \left( \displaystyle\sum_{a = 1}^{2N_P} \frac{\beta_a}{\beta} V_a \right) \ket{\psi_i} \right] \nonumber \\ &= \frac{1}{2} - \frac{1}{2\beta} \frac{\partial}{\partial \gamma_k} |\braket{\psi_i}{\psi} |^2.
\end{align}

If we want an estimator $\tilde{p}$ generated with a quantum circuit to give a good approximation of $p$ with high probability, we can use Theorem~\ref{thm:ae}. In particular, we set $k = 3$ so our success probability is greater than $1 - \frac{1}{4} = \frac{3}{4}$. We then set the $R$ in Theorem~\ref{thm:ae} such that

\begin{equation}
\frac{2 k^2 \pi^2}{R} = \frac{18 \pi^2}{R} \leq \frac{\epsilon}{2\beta},
\end{equation}
so we simply choose $R \geq \frac{36 \beta \pi^2}{\epsilon}$. Thus, by Theorem~\ref{thm:ae} and the required lower bound on $R$, with

\begin{equation}
\label{eq:rscale}
O(R) = O \left( \frac{\beta}{\epsilon} \right) = O \left( \frac{1}{\epsilon} \displaystyle\sum_{j = 1}^{N_P} \left| \frac{\partial \nu(R_i; \gamma)_j}{\partial \gamma_k} \right| \right)
\end{equation}
executions of $V$, we get an estimator $\tilde{p}$ such that $|\tilde{p} - p| < \frac{\epsilon}{2\beta}$. It follows that

\begin{multline}
    \left| 2\beta \left( \frac{1}{2} - \tilde{p} \right) - \frac{\partial}{\partial \gamma_k} |\langle \psi_i | \psi \rangle|^2 \right| \\ = \left| 2\beta \left( \frac{1}{2} - \tilde{p} \right) - 2\beta \left( \frac{1}{2} - p \right) \right| = 2\beta |\tilde{p} - p| < \epsilon,
\end{multline}
so we have an $\epsilon$-good approximation of the desired derivative with probability greater than $\frac{3}{4}$ when using the estimator $2\beta \left( \frac{1}{2} - \tilde{p} \right)$. To boost the probability of success to $1 - \delta$, we make use of a Chernoff bound. In particular, we can query our estimator $O\left(\log \left( \frac{1}{\delta} \right) \right)$ times and take the median of the outputs in order to get an approximation of the desired quantity with probability $1 - \delta$. It follows that $m_i$, the number of circuit executions required to get an $\epsilon$-good approximation with probability greater than $1 - \delta$ is, by Eq.~\eqref{eq:rscale},

\begin{equation}
    m_i \in O \left( \frac{1}{\epsilon} \log \left( \frac{1}{\delta} \right) \displaystyle\sum_{j = 1}^{N_P} \left| \frac{\partial \nu(R_i; \gamma)_j}{\partial \gamma_k} \right| \right)
\end{equation}
In order to estimate the required number of circuit executions to compute the full derivative $\frac{\partial C}{\partial \gamma_k}$, we follow the same reasoning as in Sec.~\ref{sec:had_test} and take the sum over all $m_i$, to get

\begin{equation}
    M = \displaystyle\sum_{i = 1}^{N} m_i \in O \left( \frac{N}{\epsilon} \log \left( \frac{1}{\delta} \right) \max_i \displaystyle\sum_{j = 1}^{N_P} \left| \frac{\partial \nu(R_i; \gamma)_j}{\partial \gamma_k} \right| \right).
\end{equation}

\end{proof}

As a final point, the techniques of~\cite{gilyen2017optimizing} can be used to reduce the scaling with $N$ quadratically from that given in the above result.  The central caveat with this method is that the gradient needs to be stored in a qubit register, and the quantum neural network function needs to be implemented on a quantum computer.  This leads to a prohibitive number of qubits for realistic examples and so we ignore such optimizations here as they are unlikely to be practical in their current form for our generative approach.

\section{Limitations of Generative Training}
\label{sec:4}

As was discussed in Sec.~\ref{sec:1}, producing quantum data with a quantum computer or obtaining it from an experiment is, in general, an expensive and arduous task. Therefore, understanding the amount of data required for the model to learn accurate representations of ground states fundamentally governs its feasibility. In order to derive lower bounds on the amount of data that must be supplied to the model, it is important to keep in mind that both the data and the learning procedure considered in this work are quantum in nature. As a result, we are able to draw upon a wealth of results in quantum computing and quantum information theory to make rigorous claims about the model's behaviour, and thus its performance. In the following sections, we introduce two lower bounds on the amount of data that must be provided to general quantum generative models (operating under some particular assumptions) in order for it to learn accurate models of a parameterized ground state.

The first of these lower bounds is based on Grover search. It shows that if we were able to efficiently train generic fermionic quantum models using data from one side of an avoided crossing, i.e., from a point where the energies of two states approach each other but do not cross, then we could use such an algorithm violate lower bounds on the query complexity. The second relies on interpreting a quantum generative model as a parameter-estimation problem, and invoking the quantum Cramer-Rao bound, which provides a lower bound on the variance of any unbiased estimator of the true model parameters given data sampled from the model.

\subsection{Grover Search Lower Bound}

We prove a general result showing that our algorithm will not give a universal exponential speedup. We achieve this by reducing the computation of the ground state of a fermionic Hamiltonian to Grover's search problem and then using known simulation results~\cite{kieferova2019simulating,low2018hamiltonian} and the lower bound for unstructured search~\cite{bennett1997strengths}. Specifically, we are interested in showing how much information can be learned about an unknown marked state through an adiabatic crossing, also known as an avoided crossing (where two coupled states of the same symmetry cannot cross), where quantum data can only be sampled along a coordinate before the crossing occurs. The motivation is that this often occurs in molecular energy surfaces along reaction coordinates, and we would like to know the fundamental limitation of learning a general potential energy surface when only sampling ground states nearby the equilibirum geometry. Additionally, this discussion is important in light of the performance seen in small numerical experiments discussed in Section~\ref{sec:5}.  While these results are encouraging, they raise questions about whether hard instances can emerge when applying the generative model to larger systems. We show below that hard instances exist and further are common in systems with phase transitions or conical intersections~\cite{yarkony1996diabolical}.


It was shown in Ref~\cite{roland2002quantum} that the Grover search problem can be encoded in an instance of quantum adiabatic evolution. The Hamiltonian in a 2-dimensional subspace, used for adiabatic time-evolution, is defined as 

\begin{align}
    H(s) = (1-s) H_0 + s H_m,
\end{align}
where $H_0 = \mathbbm{1} - \ket{\psi_0}\bra{\psi_0}$ and $H_m = \mathbbm{1} - \ket{m}\bra{m}$. Although $H(s)$ is distinct from the parametrized molecular Hamiltonian $H(R)$, the adiabatic parameter $s$ plays a role analogous to the coordinates $R$. We thus employ $H(s)$ as an illustration of difficulties that may arise when dealing with avoided crossings in molecular Hamiltonians, specially for the fermionic version we describe below. The state $\ket{\psi_0}$ is the superposition over all $N = 2^n$ computational basis states, $\ket{i}$, for $n$ qubits, given by
\begin{equation}
    \ket{\psi_0} = \frac{1}{\sqrt{N}} \sum_{i=0}^{N-1} \ket{i}.
\end{equation}
The single basis state $\ket{m}$ is the ``marked'' state: the single item that is to be searched for over the adiabatic evolution. The 2D subspace is represented by an orthogonalized version of the basis spanned by $\ket{\psi_0}$ and $\ket{m}$.

By modifying the Hamiltonian Grover search problem to obey fermionic statistics, we can derive a fermionic variant of the search problem with the purpose of determining how much quantum data (as a set of previously prepared ground states) over a potential energy surface with coordinate $s$ are necessary to predict the target state $\ket{m}$ within some error $\epsilon$. Finally, this lower bound can be interpreted as a worst-case scenario compared to a general avoided crossing along a nuclear coordinate for some molecular Hamiltonian. The unmarked search problem has an exponentially small, but non-zero overlap connecting the starting and end state along the coordinate of interest, and is confined to a single dimension. In general, higher-dimensional potential energy surfaces allow multiple pathways from the starting state to end state, resulting in multiple possible paths with larger minimum eigenvalue gaps.    

\begin{lemma}\label{lem:fermion-grover}
The single-particle occupied 2-dimensional subspace of the one-body fermionic Hamiltonians, $H^f_0$ and $H^f_m$ which are adiabatically connected as $ H^f(s) = (1-s)H^f_0 + sH^f_m$, are equivalent to a unary encoding of the Grover Hamiltonian $H_0$ and $H_m$ respectively.
\end{lemma}

\begin{proof}
The single-particle fermionic Hamiltonian can be defined as the following: assuming the Jordan-Wigner transformation where $n$ qubits are equivalent to $n$ fermionic sites,
\begin{align}
    H_0^f &= \mathbbm{1} - \mathcal{F}^{\dagger} \tilde{a}^{\dagger}_0 \tilde{a}_0 \mathcal{F} \\
    H_m^f &= \mathbbm{1} - a^{\dagger}_m a_m,
\end{align}
where $H_m^f$ encodes the target ``marked'' state corresponding to a single electronic configuration with one fermion occupied on site $m$, and all other sites unoccupied. Here, $\mathcal{F}$ is the fermionic fast Fourier transform (FFFT)~\cite{babbush2018low} defined as
\begin{equation}
    \mathcal{F}^{\dagger} \tilde{a}_k^{\dagger} \mathcal{F} = \frac{1}{\sqrt{n}} \sum_{p=0}^{n-1} e^{-i {\frac{2\pi}{n}} p k} \hat{a}^\dagger_p,
\end{equation}
where $\tilde{a}^{\dagger}$ denotes the creation operator over the momentum basis. Using this definition, and the fact that the  we can expand $H_0^f$ and simplify to the following form, we find that

\begin{align}
     H^f_0 &= \mathbbm{1} - \mathcal{F}^{\dagger} \tilde{a}^{\dagger}_0 \tilde{a}_0 \mathcal{F} = \mathbbm{1} - \mathcal{F}^{\dagger} \tilde{a}^{\dagger}_0\mathcal{F} \mathcal{F}^{\dagger}\tilde{a}_0 \mathcal{F} \nonumber \\ &= \mathbbm{1} - \left(\frac{1}{\sqrt{n}} \sum_{p=0}^{n-1} e^{-i {\frac{2\pi}{n}} p \cdot 0} a^\dagger_p \right) \left(\frac{1}{\sqrt{n}} \sum_{q=0}^{n-1} e^{i {\frac{2\pi}{n}} q \cdot 0} a_q \right) \nonumber \\ 
     &= \mathbbm{1} - \left(\frac{1}{\sqrt{n}} \sum_{p=0}^{n-1}  \hat{a}^\dagger_p \right) \left(\frac{1}{\sqrt{n}} \sum_{q=0}^{n-1} a_q \right) = \mathbbm{1} - \frac{1}{n} \sum_{p,q=0}^{n-1}  a^\dagger_p a_q.
\end{align}

The full Hamiltonian spanning the potential energy surface over coordinate $s$ is then in the familiar form
\begin{equation}\label{eq:adiab_ham}
    H^f(s) = (1-s)H^f_0 + sH^f_m,
\end{equation}
where $0 \leq s \leq 1$. Since these Hamiltonians conserve particle number, the chosen subspace for the starting state will retain the same symmetry varying the Hamiltonian over coordinate $s$. In the single particle case, the basis is identical to a unary encoding where the number of basis states is equal to $n$ qubits corresponding to $\ket{i} \in \{\ket{100\cdots0}, \ket{010\cdots0}, \cdots, \ket{000\cdots1} \}$. The ground state vectors for $H_0^f$ and $H_m^f$ for the single particle case are then
\begin{align}
    \ket{\phi_0} &= \frac{1}{\sqrt{n}} \sum_{i=0}^{n-1} \ket{i} \\
    \ket{m} &= \ket{00\cdots 1_m \cdots 00},
\end{align}
respectively. As we can see from this construction, the single particle $H^f (s)$ matches the construction of $H(s)$, and is therefore a unary encoding of the adiabatic Grover Hamiltonian. 
\end{proof}


Our goal is now to show that for an arbitrary quantum algorithm, the query complexity of the set of Hamiltonians over some parameter $s$, which in the case of a potential energy surface would be a bond coordinate, cannot perform better than Grover's lower bound. To illustrate this, we show a sketch of the Grover Hamiltonian potential energy surface with respect to the reduced time parameter $s$, in Fig.~\ref{fig:adiab_sketch}. Here, $E_0$ and $E_1$ are the true ground state and first excited state energy levels of $H(s)$. At $s=0$ the ground state is $| \psi_0 \rangle$, and at $s=1$, the ground state is the marked state $| m \rangle$.

\begin{figure}[t!]
    \centering
    \includegraphics[width=0.5\textwidth]{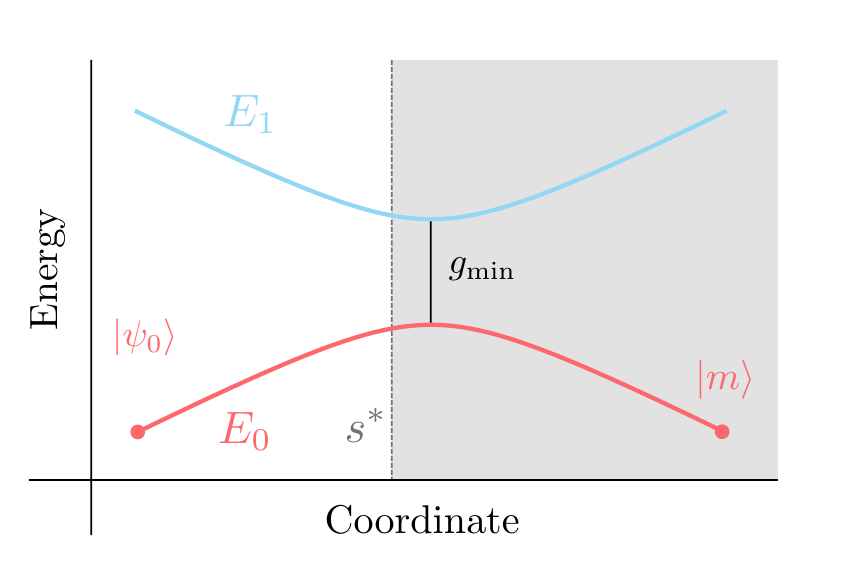}
    \caption{A schematic of the adiabatic Grover's search Hamiltonian, where $E_0$ is the ground state, $E_1$ is the first excited state of $H(s)$ over coordinate $s$, and $g_{\text{min}}$ is the location of the minimum eigenvalue gap.
    }\label{fig:adiab_sketch}
\end{figure}
As generalization is a central goal of our algorithm, one of our main objectives is to predict ground states across avoided crossings. In this case, we restrict our samples to have a maximum value of $s^*$, where the global minimum eigenvalue gap, $g_{\text{min}}$ is at $s=1/2$, and the allowed region is $s^* < 1/2$. In short, we assume we cannot take sample data from the grey shaded region in Fig.~\ref{fig:adiab_sketch}, and that $| \psi_0 \rangle$ corresponds to a molecular equilibrium geometry. With these assumptions, we then prove that no quantum algorithm can predict a target state $| m \rangle$ using fewer queries to the Hamiltonian than is allowed by Grover's lower bound. Fundamentally, this puts a lower limit on the number of required data points that is related to $s^*$ and the minimum energy gap $g_\text{min}$.

\begin{theorem}
Let $\ket{\phi}\in \mathbb{C}^{2^n}$ be the $k^{\rm th}$ eigenstate of the Hamiltonian $H_\phi$ in the list of eigenstates sorted in increasing order by eigenvalue, and let $\ket{m}$ be the corresponding eigenvector of $H_m$.  Assume the following:
\begin{enumerate}
    \item There exists a gapped adiabatic path (meaning an ordered set of sequence of Hamiltonians that have a non-zero eigenvalue gap between the ground and first excited instantaneous eigenstates) between $H_\phi$ and $H_m$ such that the instantaneous eigenstate $\ket{\psi(s)}$ for $s\in [0,1]$ has a minimum gap of at least $g_{\min}>0$ with the remainder of the spectrum.
    \item There exists $s^* \in [0,1]$ such that the spectral gap is in $\Omega(1)$ for all $s\le s^*$ as $n$ increases.
\item 
If $U_M$ is a hypothetical unitary channel such that for $M\in \tilde{o}(1/{g_{\min}}^{1/2})$, $s_j \in [0,s^*]$, ancillary state $\rho_{anc}$ and arbitrary fermionic Hamiltonian $H_\phi$, it holds that
\begin{align}\label{eq:um_channel_expect}
    {\rm Tr}\Bigg(U_M \left(\bigotimes_{j=1}^M \ket{\psi(s_j)}\otimes \rho_{anc}\otimes \ket{0}\right) \nonumber \\ (I^{\otimes M}\otimes I ) \otimes |m\rangle\!\langle m|\Bigg)\ge 2/3.
\end{align}
\end{enumerate}
Then under these assumptions if a number of samples $M\in\tilde{o}(1/g_{\min}^{1/2})$ could be used to generate the groundstate then that would violate the  $\Omega(\sqrt{2^n})$ query lower bound on the search problem.

\end{theorem}

\begin{proof}
Assume there does in fact exist a unitary channel $U_M$ that satisfies Eq.~\eqref{eq:um_channel_expect} for fermionic Hamiltonians $H_\phi$ and $H_m$.  We focus our attention on the single excitation subspace of the fermionic Hamiltonian, which is to say that our Hilbert space consists of states of the form
\begin{equation}
    \sum_{j=0}^{n-1} \psi_j a_j^\dagger\ket{0}.
\end{equation}

Using Lemma \ref{lem:fermion-grover}, we can map $H_\phi$ and $H_m$ to Grover Hamiltonians with a unary encoding of the search problem. In the case of general $H_\phi$ and $H_m$ fermionic Hamiltonians, we consider that the single-particle, one-body fermionic Grover Hamiltonian mapping provides a worst-case scenario, in that we have an exponentially small but non-zero overlap with increasing dimension $n$ of the starting state $\ket{\phi}$ and the final marked state $\ket{m}$.  As this Hamiltonian is one-body, it keeps the dynamics within the subspace.

The reduced time adiabatically connected Hamiltonian can then be defined as 
\begin{equation}
    H(s) = (1-s)H_{\phi} + s H_m,
\end{equation}
where we can then define the following from Lemma \ref{lem:fermion-grover},
\begin{align}
    H_{\phi} &\rightarrow \mathbbm{1} - \frac{1}{N}\sum_{i,j=0}^{N-1} |i\rangle\!\langle j| \\
    H_{m} &\rightarrow \mathbbm{1} -  |m\rangle\!\langle m|,
\end{align}
giving us the mapping to the adiabatic Grover Hamiltonian for the fermionic $H(s)$, where $s \in \{0,1\}$. From the adiabatic theorem~\cite{jansen2007bounds,teufel2003adiabatic,cheung2011improved} we know that the error in the prepared wavefunction, $\ket{\chi}$, after the full time interval can be written as
\begin{multline}
    | \langle \psi_k (s=1) | \chi(s=1) \rangle |^2 \\ = 1 - \left(\frac{\max_s \left| \left \langle \psi_{k+1}(s) \left| \frac{d H(s)}{ds} \right| \psi_k (s) \right \rangle \right|}{g_{\text{min}}^3T}\right)^2.
\end{multline}

 Using the first-order error bounds on the local adiabatic theorem from~\cite{jansen2007bounds}, with the requirement that $\max \{|| H_{\phi} ||, || H_m ||\} = 1$, we find that in the large $N$ limit, the total time evolution needed to achieve error $\epsilon$ in the overlap is
\begin{equation}\label{eq:T_adiab_global}
    T \in O \left(\frac{1}{g^3_{\text{min}} \epsilon} \right).
\end{equation}
Note that improved gap scaling is possible using other approaches~\cite{cheung2011improved,elgart2012note}; however, we use this scaling because it holds regardless of the relative sizes of $\epsilon$ and $g_{\rm min}$.

In order to convert this into a discrete query model to prepare the $M$ quantum states, we must first represent $H(s)$ as a linear combination of unitaries (LCU). We can define $H(s)$ in terms of unitary reflection operators from Grover's algorithm as
\begin{align}
    U_{\phi} &= \mathbbm{1} - 2 H_{\phi} \\
    U_{m} &= \mathbbm{1} - 2 H_{m}. 
\end{align}
Therefore
\begin{equation}
    H(s) = \frac{(1-s)}{2}  \left[\mathbbm{1} - U_{\phi} \right] + \frac{s}{2}  \left[\mathbbm{1} - U_{m} \right].
\end{equation}
Now, this LCU form of the Hamiltonian can be used to simulate the reduced time dynamics with the Dyson expansion method~\cite{low2018hamiltonian, kieferova2019simulating}, where the query complexity (specifically in terms of $U_m$) is in $O \left(T \frac{\log (T/\delta)}{\log \log (T/\delta)} \right)$, where $\delta$ is the simulation error. Using Eq.~\eqref{eq:T_adiab_global}, we can then show that the query complexity to $U_m$ is in 
\begin{equation}
    O \left( \frac{1}{g^3_{\text{min}} \epsilon} \frac{\log (1/(g^3_{\text{min}} \epsilon \delta))}{\log \log (1/(g^3_{\text{min}} \epsilon \delta ))} \right).
\end{equation}

In order to understand the error in the samples prepared up to the imposed boundary of $s^*$, we can then define an $s^*$-rescaled Hamiltonian as
\begin{equation}
    H^*(s) = \alpha(s) H_0 + \beta(s) H_m.
\end{equation}
The new boundary conditions of $H^*(s)$ can then be derived using the following equations, to solve for both $\alpha(s)$ and $\beta(s)$. First, 
\begin{align}
    H^*(0) &= H_0\\
    H^*(1) &= (1-s^*)H_0 + s^* H_m.
\end{align}
Therefore,
\begin{align}
    \alpha(0) &= 1 \, , \,\, \alpha(1) = 1-s^*  \\
    \beta(0) &= 0 \, , \,\, \beta(1) = s^*,
\end{align}
and the rescaled Hamiltonian is then
\begin{equation}
    H^*(s) = (1-ss^*) H_0 + ss^* H_m.
\end{equation}

For $H^*(s)$, we now have a separate adiabatic state preparation error distinct from $\epsilon$, since the derivative term is now
\begin{equation}
    \frac{d H^*(s)}{ds} = -s^* H_0 + s^* H_m = s^* \frac{d H(s)}{ds}.
\end{equation}
The adiabatic error in $H^*(s)$ is then
\begin{equation}
    | \langle \psi_k (s_j) | \chi(s_j) \rangle |^2 = 1 - \gamma^2_j,
\end{equation}
for the $j$th prepared sample wavefunction at coordinate $s_j$. Using this fact, this rescaled error term then scales in a similar fashion to the full $H(s)$, using the local adiabatic evolution as
\begin{equation}
    T_j^* \in O \left(\frac{1}{g^{*3}_{\text{min}} \gamma_j} \right),
\end{equation}
for every $j$th sample, where $g_{\text{min}}^{*}$ is the minimum eigenvalue gap of $H(s)$ where $s \leq s*$. Rearranging for total time evolution, and dropping the index $j$ off of $\gamma_j$ for simplicity, we get that the total time required to prepare $M$ states is 
\begin{equation}\label{eq:T_adiab_local}
    T^* \in O \left(\frac{M}{g^{*3}_{\text{min}} \gamma} \right).
\end{equation}
Converting this into the query complexity of $U_{m}^*$ using the same Dyson expansion method as above, we get
\begin{equation}
    O \left( \frac{M}{g_{\text{min}}^{*3} \gamma} \frac{\log (1/g_{\text{min}}^{*3} \gamma \delta)}{\log \log (1/g^{*3}_{\text{min}}\gamma \delta )} \right) .
\end{equation}
Now, combining the error terms into $\epsilon' = \delta + \gamma$ we can define the $j$th evolved state as 
\begin{equation}
    | \tilde{\psi}(s_j) \rangle = | \psi(s_j) \rangle + \epsilon' \ket{E_j},
\end{equation}
with error $O(\epsilon')$, and $\ket{E_j}$ is defined as the error vector associated with $\epsilon'$. In order to understand how the error term modifies the expectation value in Eq.~\eqref{eq:um_channel_expect}, we substitute in $| \tilde{\psi}(s_j) \rangle$ and get the following relation utilizing the von-Neumann trace inequality
\begin{align}
    {\rm Tr}&\Bigg( \Bigg(\bigotimes_{j=1}^M | \tilde{\psi}(s_j) \rangle \langle \tilde{\psi}(s_j) | - | \psi(s_j) \rangle \langle \psi(s_j) | \nonumber \\ &\otimes \rho_{anc}\otimes |0 \rangle \langle 0 |\Bigg) U_M^{\dagger} \left( I^{\otimes M}\otimes I \otimes |m\rangle\!\langle m| \right) U_M\Bigg) \nonumber \\
    &\le \left\| \bigotimes_{j=1}^M | \tilde{\psi}(s_j) \rangle \langle \tilde{\psi}(s_j) | - | \psi(s_j) \rangle \langle \psi(s_j) | \right\|_1 \, .
\end{align}
In terms of the true eigenstate and the error vector we then have
\begin{align}
    &\left\| \bigotimes_{j=1}^M | \tilde{\psi}(s_j) \rangle \langle \tilde{\psi}(s_j) | - | \psi(s_j) \rangle \langle \psi(s_j) | \right\|_1 \nonumber \\ &\le M \epsilon' \left\| | \psi(s_j) \rangle \langle E_j | + | E_j \rangle \langle \psi(s_j) | + \epsilon'| E_j \rangle \langle E_j | \right\|_1 \nonumber \\
    & \le M \epsilon' (2 + \epsilon'),
\end{align}
and if $\epsilon' \le 1$ then 
\begin{align}
     M \epsilon' (2 + \epsilon') &\le 3M\epsilon'.
\end{align}
If we let $3M\epsilon' = 1/6$, with the motivation that we bound worst case expectation value of Eq.~\eqref{eq:um_channel_expect} to be $1/2$, we see that
\begin{equation}
    \epsilon' = \frac{1}{18M}.
\end{equation}
Using this relationship we then see that the query complexity for preparing $M$ states with max error $\epsilon'$ is 
\begin{equation}
    O \left( \frac{M^2}{g_{\text{min}}^{*3}} \frac{\log (1/g_{\text{min}}^{*3} )}{\log \log (1/g^{*3}_{\text{min}} )} \right) = \widetilde{O}\left( \frac{M^2}{g_{\text{min}}^{*3}} \right).
\end{equation}

By substituting in $M^2 = O(g_{\text{min}}^{*3} \sqrt{2^n})$, we produce a contradiction of the bounds in~\cite{boyer1998tight} 
because of the assumption that
an algorithm exists such that
\begin{align}
    {\rm Tr}\Bigg(U_M \left(\bigotimes_{j=1}^M \ket{\psi(s_j)}\otimes \rho_{anc}\otimes \ket{0}\right) \nonumber \\ (I^{\otimes M}\otimes I ) \otimes |m\rangle\!\langle m|\Bigg)\ge 2/3.
\end{align}  

This would be implied by $M\in \widetilde{o}(1/g_{\min}^{1/2})$ which contradicts Assumption 3.

\end{proof}
This shows that the structure of the phase diagram of a quantum system places fundamental limitations in our ability to extrapolate beyond the training set.  In particular, it reveals that avoided crossings in the spectrum can produce intractable barriers in our ability to predict the groundstate in different configurations. We will see below in Sec.~\ref{sec:5} that these barriers naturally appear as we try to extrapolate towards an eigenvalue crossing as we approach dissociation for molecules.

\subsection{Quantum Cramer-Rao Lower Bound}

The entire procedure executed by the generative model can be summarized as performing measurements of the supplied data states $|\psi_0(R_i)\rangle$, doing classical post-processing on the results, then using these results to perform further measurements on new copies of $|\psi_0(R_i)\rangle$, repeating this procedure until some convergence criterion is met. It follows that we can consider the optimization and parameter-updates taking place in the model as together forming one large entangling measurement on a state of the form

\begin{equation}
    |\psi_{\text{data}}\rangle = |\psi_0(R_1)\rangle \otimes \cdots \otimes |\psi_0(R_1)\rangle \otimes \cdots \otimes |\psi_0(R_N)\rangle \otimes |\psi_0(R_N)\rangle,
\end{equation}
where $|\psi_0(R_i)\rangle$ is repeated $m_i$ times in the tensor product, along with an ancilla register $|0\rangle^{\otimes a}$. Let us make the assumption that we have chosen the circuit $U$ and the neural network $\nu$ such that there exists some set of neural network parameters $\tilde{\gamma}$ where 

\begin{equation}
|\psi(\nu(R_i; \tilde{\gamma}))\rangle = U(\nu(R_i, \tilde{\gamma}))|\psi_{\text{init}}(R_i)\rangle = |\psi_0(R_i)\rangle
\end{equation}
for each $R_i$. We can therefore think of $|\psi_{\text{data}}\rangle$ as a state in which the optimal (but unknown) parameters $\tilde{\gamma}$ are embedded. Since our generative model outputs a set of parameters $\gamma^{*}$, where the goal is for $|\psi(\nu( R_i; \gamma^{*}))\rangle \approx |\psi_0(R_i)\rangle = |\psi(\nu(R_i; \tilde{\gamma}))\rangle$ for all $R_i$, we can think of the model as performing parameter-estimation: trying to find a good estimate of the optimal parameters $\tilde{\gamma}$ via measurement of $|\psi_{\text{data}}\rangle$. Parameter-estimation problems have been widely studied in the context of quantum information, and there are many results which give bounds on the performance of any such estimator under particular conditions \cite{helstrom1969quantum}. We make use of a theorem known as the quantum Cramer-Rao bound in order to make claims about the number of data states required to minimize the infidelity cost function $C$ of Eq.~\eqref{eq:cost} with high probability.

Before formally stating the quantum Cramer-Rao bound, we must introduce the concept of the quantum Fisher information matrix (QFI) for the (classical) neural network parameters $\gamma$. The following definition of the QFI is specific for pure states and is shown to be equivalent to the standard expression in this case in Theorem 2.5 of~\cite{liu2019quantum}.

\begin{definition}[Quantum Fisher Information Matrix]
\label{def:qfi}
The quantum Fisher information (QFI) matrix of a parameterized quantum state $|\psi(\gamma)\rangle$ at $\gamma = \gamma_0 \in \mathbb{R}^{n}$ is the matrix $\mathcal{F}[|\psi(\gamma_0)\rangle]$ with entries

\begin{align}
    \mathcal{F} \left[ |\psi(\gamma_0)\rangle \right]_{ij} &= 4 \mathrm{Re} \Big[ \langle \partial_i \psi(\gamma_0) | \partial_j \psi(\gamma_0) \rangle \nonumber \\ & - \langle \psi(\gamma_0) | \partial_j \psi(\gamma_0) \rangle \langle \partial_i \psi(\gamma_0) | \psi(\gamma_0) \rangle \Big]
    \label{eq:nqfi}.
\end{align}
\end{definition}

With this equivalent definition of the quantum Fisher information, we can compute upper bounds on the entries of the QFI matrix of a state parametrized by the neural network parameters $\gamma$, given by $|\psi(\gamma)\rangle := |\psi(\nu(R; \gamma))\rangle$, for some arbitrary $R$. Recall that

\begin{align}
    |\psi(\nu(R; \gamma)\rangle &= U(\nu(R; \gamma))|\psi_{\text{init}}(R)\rangle \nonumber \\ &= \left( \displaystyle\prod_{j = 1}^{N_P} e^{-i \nu(R; \gamma)_j H_j} \right) |\psi_{\text{init}}(R)\rangle.
\end{align}

Using an identical procedure used to derive Eq.~\eqref{eq:gr}, we find that

\begin{equation}
    \partial_j | \psi(\gamma) \rangle = -i \displaystyle\sum_{k = 1}^{N_P} \frac{\partial \nu(R; \gamma)_{k}}{\partial \gamma_j} \hat{H}_k(\nu(R; \gamma)) |\psi(\gamma)\rangle,
\end{equation}

where, as before, we define

\begin{equation}
\hat{H}_k(\gamma) = \left( \prod_{p < k} e^{-i \gamma_p H_p} \right) H_k \left( \prod_{p < k} e^{-i \gamma_p H_p} \right)^{\dagger}.
\end{equation}
It then follows that

\begin{align}
    &\langle \partial_i \psi(\gamma) | \partial_j \psi(\gamma) \rangle = \displaystyle\sum_{k, \ell} \frac{\partial \nu(R; \gamma)_{k}}{\partial \gamma_i} \frac{\partial \nu(R; \gamma)_{\ell}}{\partial \gamma_j} \nonumber \\ & \times \langle \psi(\gamma) | \hat{H}_k(\nu(R; \gamma)) \hat{H}_{\ell}(\nu(R; \gamma)) |\psi(\gamma)\rangle
\end{align}

and

\begin{align}
    &\langle \psi(\gamma) | \partial_j \psi(\gamma) \rangle \langle \partial_i \psi(\gamma) | \psi(\gamma) \rangle = \displaystyle\sum_{k, \ell} \frac{\partial \nu(R; \gamma)_{k}}{\partial \gamma_i} \frac{\partial \nu(R; \gamma)_{\ell}}{\partial \gamma_j} \nonumber \\ &\times \langle \psi(\gamma) | \hat{H}_k(\nu(R; \gamma)) | \psi(\gamma)\rangle \langle \psi(\gamma)| \hat{H}_{\ell}(\nu(R; \gamma)) |\psi(\gamma)\rangle.
\end{align}

Thus, making use of Eq.~\eqref{eq:nqfi}, it follows that

\begin{align}
    &\mathcal{F}[|\psi(\gamma)\rangle]_{ij} = 4 \displaystyle\sum_{k, \ell} \frac{\partial \nu(R; \gamma)_k}{\partial \gamma_i} \frac{\partial \nu(R; \gamma)}{\partial \gamma_j} \Big[ \langle \psi(\gamma) | \hat{H}_k \hat{H}_{\ell} |\psi(\gamma)\rangle \nonumber \\ &- \langle \psi(\gamma) | \hat{H}_k | \psi(\gamma)\rangle \langle \psi(\gamma)| \hat{H}_{\ell} |\psi(\gamma)\rangle \Big] \nonumber \\ &= 4 [ \partial_i \nu(R; \gamma) ]^{T} \text{Cov}(\hat{H}, \hat{H}) [ \partial_j \nu(R; \gamma) ] \nonumber \\ &\leq 4 || \partial_i \nu(R; \gamma) ||_2 || \partial_j \nu(R; \gamma) ||_2 || \text{Cov}(\hat{H}, \hat{H}) ||_2 \nonumber \\ & \leq 4 N_P^{1/2} || \partial_i \nu(R; \gamma) ||_2 || \partial_j \nu(R; \gamma) ||_2 || \text{Cov}(\hat{H}, \hat{H}) ||_{1, 1},
    \label{eq:qcrlb}
\end{align}
where $\|\cdot\|_2$ is the (induced) Euclidean norm and $\text{Cov}(\hat{H}, \hat{H})$ is the matrix such that

\begin{align}
\text{Cov}(\hat{H}, \hat{H})_{k\ell} &= \langle \psi(\gamma) | \hat{H}_k \hat{H}_{\ell} |\psi(\gamma)\rangle \nonumber \\ &- \langle \psi(\gamma) | \hat{H}_k | \psi(\gamma)\rangle \langle \psi(\gamma)| \hat{H}_{\ell} |\psi(\gamma)\rangle,
\end{align}
and $|| \cdot ||_{1, 1}$ is the entry-wise matrix norm, which is precisely the sum of the absolute values of each of the matrix entries. It follows that $\frac{|| \text{Cov}(\hat{H}, \hat{H}) ||_{1, 1}}{N_P^2}$ is precisely the average of the absolute values of the ``covariances'' between the pairs of generators in the circuit ansatz. Going forward, we denote this quantity by $\overline{\text{Cov}(\hat{H}, \hat{H})}$. The bound of Eq.~\ref{eq:qcrlb} becomes

\begin{align}
&\mathcal{F}[|\psi(\gamma)\rangle]_{ij} = \mathcal{F}[|\psi(\nu(\gamma; R))\rangle]_{ij} \nonumber \\ &\leq 4 N_P^{5/2} || \partial_i \nu(R; \gamma) ||_2 || \partial_j \nu(R; \gamma) ||_2 \overline{\text{Cov}(\hat{H}, \hat{H})}
\end{align}

Now, using Definition~\ref{def:qfi}, we can state the result known as the quantum Cramer-Rao bound.

\begin{theorem}[Quantum Cramer-Rao Bound]
\label{thm:qcr}
Given a parameterized quantum state $|\psi(\gamma)\rangle$ that is a function of a hidden vector of parameters $\gamma$ and a differentiable function $\widehat{T}(\gamma)$ which is an estimator of $\gamma$ obtained from measuring $|\psi(\gamma)\rangle$, the covariance matrix of $\widehat{T}(\gamma)$ satisfies the following matrix inequality~\cite{meyer2021fisher}:

\begin{equation}
    {\rm Cov} \big( \widehat{T}(\gamma) \big) \geq Df(\gamma) \mathcal{F} \left[ |\psi(\gamma)\rangle \right]^{-1} Df(\gamma)^{T}
\end{equation}
where $f(\gamma) = \mathbb{E}[\widehat{T}(\gamma)]$ and we use the convention that $A\ge B$ for matrices $A$ and $B$ implies that $A - B$ is a positive semi-definite matrix and $Df(\gamma)$ is the Jacobian matrix for $f$ at $\gamma$.
\end{theorem}

Most references state the quantum Cramer-Rao bound only for the case of an unbiased estimator, but the more general case of the variance of a biased estimator satisfying the given matrix inequality follows directly from Equation 149 of Ref.~\cite{meyer2021fisher}. Going forward, we write an estimator of parameters $\gamma$ as $\widehat{\gamma}$, keeping the dependence on $\gamma$ implicit rather than explicitly writing the estimator as a function of $\gamma$. Furthermore, we always assume an estimator is a differentiable function of the true parameters it is attempting to estimate. This is a strong assumption which will be discussed later. Further, let $C(\gamma)$ be the infidelity loss function for our generative model given in Eq.~\ref{eq:cost}, which we restate for convenience:

\begin{equation}
    C(\gamma) := 1 - \frac{1}{N} \displaystyle\sum_{j = 1}^{N} | \langle \psi(\nu(R_j ; \tilde{\gamma})) | \psi(\nu(R_j; \gamma)) \rangle|^2.
    \label{eq:new_cost}
\end{equation}

We also define a two-parameter generalization of $C$, which will be useful in the proof that follows:

\begin{equation}
    C(\gamma, \sigma) := 1 - \frac{1}{N} \displaystyle\sum_{j = 1}^{N} | \langle \psi(\nu(R_j; \gamma)) | \psi(\nu(R_j; \sigma)) \rangle|^2.
    \label{eq:generalized_cost}
\end{equation}

Intuitively, this function compares the ground state models induced by the neural network at two sets of neural network parameters, over the training set $\{R_i\}$. Clearly, fixing $\sigma = \tilde{\gamma}$ yields the original infidelity cost function in Eq.~\ref{eq:new_cost}.

We can now introduce our main result, which gives a lower bound, in particular situations, on the number of data states needed in order to create an estimator $\widehat{\gamma}$ of the optimal parameters $\tilde{\gamma}$ such that $C(\widehat{\gamma}) < \epsilon$ with probability $1 - \delta$, assuming that $\widehat{\gamma}$ is obtained from measurements of data states $|\psi_0(R_i)\rangle$.

\begin{theorem}
\label{thm:qcr_bnd}
Let $U$ be a circuit ansatz, let $|\psi(\nu(R; \gamma))\rangle = U(\nu(R; \gamma))|\psi_0\rangle$ where $\gamma\in \mathbb{R}^{M_P}$ is the trainable vector of neural network parameters, $R$ is the fixed geometry of the molecule and $\nu(R; \gamma)\in \mathbb{R}^{N_P}$ is the classical neural network function with $|\psi_0(R)\rangle$ being the ground state of the parameterized Hamiltonian $H(R)$. Suppose we are given access to a collection of parameters $\{R_i\}_{i = 1}^{N}$, and a state of the form

\begin{align}
    |\psi_{\text{data}}\rangle \otimes |0\rangle^{\otimes a} &= |\psi_0(R_1)\rangle \otimes \cdots \otimes |\psi_0(R_1)\rangle \nonumber \\ &\otimes \cdots \otimes |\psi_0(R_N)\rangle \otimes \cdots |\psi_0(R_N)\rangle \otimes |0\rangle^{\otimes a}
\end{align}
where the state $|\psi_0(R_i)\rangle$ is repeated $m_i$ times in the tensor product. Suppose there exist parameters $\tilde{\gamma}$ such that $|\psi(\nu(R_i; \gamma))\rangle = |\psi_0(R_i)\rangle$ for each $R_i$. Suppose we have $\hat{\gamma}$: an estimator of $\tilde{\gamma}$ supported in a convex set $A$ around $\gamma^{*} = \mathbb{E}[\hat{\gamma}]$ on which $C(\gamma)$ is strongly convex and where $\gamma^{*}$ is a local minimum of $C$, such that it is constructed by performing a measurement on $|\psi_{\text{data}}\rangle \otimes |0\rangle^{\otimes a}$. Moreover, assume that $C(\hat{\gamma}) < \epsilon$ with probability at least $1 - \delta$ over the convex set. Then, the total number of states in the tensor product defining $|\psi_{\text{data}}\rangle$, $M = \sum_{i} m_i$, satisfies the following lower bound:

\begin{align}
    &M \geq \frac{|| [\partial_{12} C](\tilde{\gamma}, \gamma^{*})||_2^2}{8 N_P^{7/2} N \overline{\mathrm{Cov}(\hat{H}, \hat{H})} (\epsilon + \delta - 2 \epsilon \delta)} \nonumber \\ &\left( \overline{|| D \nu(R_j ; \tilde{\gamma})||_2^2} \max_{\phi \in A} || C''(\phi)^{-1} ||_2 ||C''(\gamma^{*})||_2^2 \right)^{-1},
\end{align}
where $\overline{|| D \nu(R_j ; \tilde{\gamma})||_2^2}$ is the average of the derivative magnitudes $|| D \nu(R_j ; \tilde{\gamma})||_2^2$ over $j$, $C''$ is the Hessian matrix of the infidelity cost function in Eq.~\eqref{eq:new_cost}, and $\partial_{12} C$ is the mixed second-derivative matrix of the function defined in Eq.~\eqref{eq:generalized_cost}, with element $(i, j)$ given by $\frac{\partial^2 C(\gamma, \sigma)}{\partial \gamma_i \partial \sigma_j}$.

\end{theorem}

\begin{proof}
Following immediately from the definition in Eq.~\eqref{eq:new_cost}, $0 \leq C(\gamma) \leq 1$ for any $\gamma$. Thus, by our initial assumptions,

\begin{equation}
\label{eq:prob_bound}
     \mathbb{E}[C(\hat{\gamma})] \leq \epsilon (1 - \delta) + \delta.
\end{equation}

Note that $\partial_i C(\gamma^{*}) = 0$ for all $i$ since $\gamma^{*}$ is a local minimum of $C$. Thus, by the mean-value statement of Taylor's remainder theorem, we have for $\gamma \in A$, a convex set containing $\gamma^{*}$,

\begin{align}
    C(\gamma) &= C(\gamma^{*}) + \frac{1}{2} \displaystyle\sum_{i, j} [\partial_{ij} C](\zeta(\gamma))(\gamma_i - \gamma^{*}_i)(\gamma_j - \gamma^{*}_j) \nonumber \\ & = C(\gamma^{*}) + \frac{1}{2} \langle \Delta(\gamma) | C''(\zeta(\gamma)) | \Delta(\gamma) \rangle,
\end{align}
where $\zeta(\gamma)$ is a point on the straight line extending from $\gamma$ to $\gamma^{*}$, $C''(\phi)$ is the matrix such that $C''(\phi)_{ij} = [\partial_{ij} C](\phi)$, and $\Delta(\gamma)$ is the vector such that $\Delta(\gamma)_i = \gamma_i - \gamma^{*}_i$. We know that $C''(\phi)$ is invertible as we have assumed that $C$ is strongly convex on $A$. This implies that $C''(\phi)$ is positive definite on $A$. We then note that

\begin{align}
   &\frac{1}{2} \langle \Delta(\gamma) | C''(\zeta(\gamma)) | \Delta(\gamma) \rangle \geq \frac{1}{2} \sigma_{\text{min}}(C''(\zeta(\gamma))) \langle \Delta(\gamma) | \Delta(\gamma) \rangle \nonumber \\ &= \frac{\langle \Delta(\gamma) | \Delta(\gamma) \rangle}{2 || C''(\zeta(\gamma))^{-1} ||_2} \geq \frac{\langle \Delta(\gamma) | \Delta(\gamma) \rangle}{2 \max_{\phi \in A} || C''(\phi)^{-1} ||_2}.
\end{align}

It follows immediately that

\begin{align}
    \label{eq:exp_lb}
    \mathbb{E}[C(\hat{\gamma})] &\geq C(\gamma^{*}) + \frac{1}{2 \max_{\phi} || C''(\phi)^{-1} ||_2} \mathbb{E} \left[ \langle \Delta(\hat{\gamma}) | \Delta(\hat{\gamma}) \rangle \right] \nonumber \\ &= C(\gamma^{*}) + \frac{1}{2 \max_{\phi} || C''(\phi)^{-1} ||_2} \text{Tr} \left[ \text{Cov}( \hat{\gamma} ) \right].
\end{align}
Now, we can extend our lower bound on $\mathbb{E}[C(\hat{\gamma})]$ using the quantum Cramer-Rao bound (Theorem~\ref{thm:qcr}). Since the parameters we wish to estimate, $\tilde{\gamma}$, are embedded in $|\psi_{\text{data}}\rangle \otimes |0\rangle^{\otimes a}$, we can treat $\hat{\gamma}$ as an estimator of $\tilde{\gamma}$ induced from measurement of $|\psi_{\text{data}}\rangle \otimes |0\rangle^{\otimes a}$. Thus, $\hat{\gamma}$ satisfies the quantum Cramer-Rao bound. We immediately get

\begin{align}
    &\text{Tr} \left[ \text{Cov} (\hat{\gamma}) \right] \geq \text{Tr} \left( Df(\tilde{\gamma}) \mathcal{F} \left[ |\psi_{\text{data}}\rangle \otimes |0\rangle^{\otimes a} \right]^{-1} Df(\tilde{\gamma})^{T} \right) \nonumber \\ &= \text{Tr} \Big( Df(\tilde{\gamma}) \mathcal{F} \Big[ |\psi_0(R_1)\rangle \otimes \cdots |\psi_0(R_1)\rangle \nonumber\\ &\otimes \cdots \otimes |\psi_0((R_N)\rangle \otimes \cdots \otimes |\psi_0(R_N) \rangle \otimes |0\rangle^{\otimes a} \rangle \Big]^{-1} Df(\tilde{\gamma})^{T} \Big) \nonumber \\ &= \text{Tr} \left( Df(\tilde{\gamma}) \left( \displaystyle\sum_{j = 1}^{N} m_j \mathcal{F}[ |\psi_0(R_j)\rangle] \right)^{-1} Df(\tilde{\gamma})^{T} \right) \nonumber \\ &= \displaystyle\sum_{i = 1}^{N_P} [D f(\tilde{\gamma})]_i \left( \displaystyle\sum_{j = 1}^{N} m_j \mathcal{F}[ |\psi_0(R_j)\rangle] \right)^{-1} [D f(\tilde{\gamma})]_i^{T},
    \label{eq:f1}
\end{align}
where $\mathcal{F}[|\psi_0(R_j)\rangle] := \mathcal{F}[|\psi(\nu(R_j ; \tilde{\gamma}))\rangle]$. We use the additive property of the quantum Fisher information matrix over tensor products, and the fact that the quantum Fisher information of the fixed state $|0\rangle^{\otimes a}$ vanishes \cite{meyer2021fisher}.  Note that the function $f(\tilde{\gamma})$ is precisely the expectation value of the estimator $\mathbb{E}[\hat{\gamma}]$ as a function of $\tilde{\gamma}$. $[Df(\tilde{\gamma})]_i$ denotes the $i$-th row of the matrix $Df(\tilde{\gamma})$, $||\cdot ||_{F}$ is the Frobenius norm and $\|\cdot\|_2$ is the induced $2$-norm. Now, since each matrix $\mathcal{F}[|\psi_0(R_j)\rangle]$ is positive semi-definite \cite{meyer2021fisher}, and each $m_j$ is positive, it follows that the inverse $\left( \sum_{j = 1}^{N} m_j \mathcal{F}[|\psi_0(R_j)\rangle] \right)^{-1}$ is also positive semi-definite. It follows that

\begin{align}
    &\displaystyle\sum_{i = 1}^{N_P} [D f(\tilde{\gamma})]_i \left( \displaystyle\sum_{j = 1}^{N} m_j \mathcal{F}[ |\psi_0(R_j)\rangle] \right)^{-1} [D f(\tilde{\gamma})]_i^{T} \nonumber \\ &\geq \sigma_{\text{min}} \left( \left( \displaystyle\sum_{j = 1}^{N} m_j \mathcal{F}[ |\psi_0(R_j)\rangle] \right)^{-1} \right) || Df(\tilde{\gamma}) ||_F^{2} \nonumber \\ &= \left| \left| \displaystyle\sum_{j = 1}^{N} m_j \mathcal{F}[ |\psi_0(R_j)\rangle] \right| \right|_2^{-1} || Df(\tilde{\gamma}) ||_F^{2} \nonumber \\ &\geq \left| \left| \displaystyle\sum_{j = 1}^{N} m_j \mathcal{F}[ |\psi_0(R_j)\rangle] \right| \right|_2^{-1} || Df(\tilde{\gamma}) ||_2^{2}.
    \label{eq:f2}
\end{align}

Once again using the fact that the linear combination $\sum_{j = 1}^{N} m_j \mathcal{F}[|\psi_0(R_j)\rangle]$ is positive semi-definite, it follows that for some normalized $|x\rangle$,

\begin{align}
    &\Big\langle x \Big| \left[ \displaystyle\sum_{j = 1}^{N} m_j \mathcal{F}[|\psi_0(R_j)\rangle] \right] \Big| x \Big\rangle = \displaystyle\sum_{j = 1}^{N} m_j \langle x | \mathcal{F}[ |\psi_0(R_j)\rangle] |x\rangle \nonumber \\ &\leq \left( \displaystyle\sum_{j = 1}^{N} m_j \right) \left( \displaystyle\sum_{j = 1}^{N} \langle x | \mathcal{F}[|\psi_0(R_j)\rangle] | x \rangle \right) \nonumber \\    &= M \Big\langle x \Big| \left[ \displaystyle\sum_{j = 1}^{N} \mathcal{F}[|\psi_0(R_j)\rangle] \right] \Big| x \Big\rangle,
    \label{eq:f3}
\end{align}

which implies that

\begin{equation}
    \left| \left| \displaystyle\sum_{j = 1}^{N} m_j \mathcal{F}[ |\psi_0(R_j)\rangle] \right| \right|_2 \leq M \left| \left| \displaystyle\sum_{j = 1}^{N} \mathcal{F}[ |\psi_0(R_j)\rangle] \right| \right|_2.
    \label{eq:f4}
\end{equation}

It follows immediately from Eq.~\ref{eq:f1}, Eq.~\ref{eq:f2}, and Eq.~\ref{eq:f4} that

\begin{align}
    \label{eq:tr_lb}
    \text{Tr} \left[ \text{Cov}(\hat{\gamma}) \right] &\geq || Df(\tilde{\gamma})||_2^2 \left| \left| \displaystyle\sum_{j = 1}^{N} m_j \mathcal{F}[ |\psi_0(R_j) \rangle] \right| \right|_2^{-1} \nonumber \\ &\geq \frac{|| Df(\tilde{\gamma}) ||_2^2}{M} \left| \left| \displaystyle\sum_{j = 1}^{N} \mathcal{F}[ |\psi_0(R_j) \rangle] \right| \right|_2^{-1}.
\end{align}

We can now make use of the upper bound previously derived on the entries of the quantum Fisher information matrix for each of the data states in Eq.~\eqref{eq:qcrlb}. It follows again from the fact that each $\mathcal{F}[|\psi_0(R_j)\rangle]$ is positive semi-definite that

\begin{align}
    \Big| \Big| \displaystyle\sum_{j = 1}^{N} \mathcal{F}[|\psi_0(R_j)\rangle] \Big| \Big|_2 &\leq \displaystyle\sum_{j = 1}^{N} \text{Tr} \left( \mathcal{F} \left[ |\psi_0(R_j)\rangle \right] \right) \nonumber \\ &\leq 4 N_P^{5/2} \displaystyle\sum_{j = 1}^{N} \displaystyle\sum_{i = 1}^{M_P} || \partial_i \nu(R_j ; \tilde{\gamma}) ||_2^{2} \overline{\text{Cov}(\hat{H}, \hat{H})} \nonumber \\ &= 4 N_P^{5/2} \displaystyle\sum_{j = 1}^{N} || D \nu(R_j ; \tilde{\gamma}) ||_F^{2} \overline{\text{Cov}(\hat{H}, \hat{H})} \nonumber \\ & \leq 4 N_P^{7/2} N \overline{||D\nu(R_j ; \tilde{\gamma})||^2_2} \overline{\text{Cov}(\hat{H}, \hat{H})},
\end{align}
where $\overline{||D\nu(R_j ; \tilde{\gamma})||^2_2}$ is the average over $R_j$ of the squared $2$-norms of the neural network derivatives. Thus, we get, from Eq.~\ref{eq:tr_lb},

\begin{equation}
    \text{Tr} \left[ \text{Cov}(\hat{\gamma}) \right] \geq \frac{|| Df(\tilde{\gamma}) ||_2^2 }{4 M N_P^{7/2} N \overline{\text{Cov}(\hat{H}, \hat{H})} \overline{||D\nu(R_j ; \tilde{\gamma})||^2_2}},
\end{equation}
and from Eq.~\ref{eq:exp_lb},

\begin{align}
&\mathbb{E}[C(\hat{\gamma})] - C(\gamma^{*}) \geq \nonumber \\ & \frac{|| Df(\tilde{\gamma}) ||_2^2}{8 M N_P^{7/2} N \overline{\text{Cov}(\hat{H}, \hat{H})} \overline{||D\nu(R_j ; \tilde{\gamma})||_2^2} \max_{\phi \in A} || C''(\phi)^{-1} ||_2}.
\label{eq:next_bound}
\end{align}
Using Eq.~\eqref{eq:prob_bound}, it follows that

\begin{equation}
\mathbb{E}[C(\hat{\gamma})] - C(\gamma^{*}) \leq \mathbb{E}[C(\hat{\gamma})] \leq \epsilon(1 - \delta) + \delta = \epsilon + \delta - \epsilon \delta,
\end{equation}
which allows us to use Eq.~\eqref{eq:next_bound} to re-arrange for a lower bound on $M$,

\begin{align}
    M &\geq \frac{||D f(\tilde{\gamma}) ||_2^2}{8 N_P^{7/2} N \overline{\text{Cov}(\hat{H}, \hat{H})} (\epsilon + \delta - \epsilon \delta)} \nonumber \\ & \left( \overline{\|D \nu(R_j ; \tilde{\gamma})||_2^2} \max_{\phi \in A} || C''(\phi)^{-1} ||_2 \right)^{-1}.
\end{align}

To conclude the proof, we turn our attention to the function $f$. Recall the $2N_P$-variable generalization of the infidelity cost function, $C(\gamma, \sigma)$, defined in Eq.~\ref{eq:generalized_cost}. By definition, we have for any $i$ with $1 \leq i \leq M_P$, $[\partial_i C](f(\gamma), \gamma) = 0$ in a neighbourhood around $\tilde{\gamma}$, as the parameters $f(\gamma)$ are, by definition, a local optimum of the infidelity cost function of Eq.~\eqref{eq:new_cost}. This implies that

\begin{align}
    \frac{\partial}{\partial \gamma_j} [\partial_i C](f(\gamma), \gamma) |_{\gamma = \tilde{\gamma}} &= \displaystyle\sum_{k = 1}^{M_P} [\partial_k \partial_i C](f(\tilde{\gamma}), \tilde{\gamma}) \frac{\partial f(\tilde{\gamma})_k}{\partial \gamma_j} \nonumber \\ &+ [\partial_{M_P + j} \partial_i C](f(\tilde{\gamma}), \tilde{\gamma}) = 0.
\end{align}
This in turn implies that

\begin{equation}
    C''(f(\tilde{\gamma})) Df(\tilde{\gamma}) = -[\partial_{12} C](f(\tilde{\gamma}), \tilde{\gamma}),
\end{equation}
where $[\partial_{12} C](\gamma, \sigma)$ is the matrix of mixed second-order partial derivatives with respect to the parameters $\gamma$ and $\sigma$. In other words, $[\partial_{12} C](\gamma, \sigma)_{ij} = \frac{\partial^2 C(\gamma, \sigma)}{\partial \gamma_i \partial \sigma_j}$. This result allows us to lower bound $|| Df(\tilde{\gamma}) ||_2$ as

\begin{align}
    || Df(\tilde{\gamma})||_2 &= || C''(f(\tilde{\gamma}))^{-1} [\partial_{12} C](f(\tilde{\gamma}), \tilde{\gamma})||_2 \nonumber \\ & \geq \frac{|| [\partial_{12} C](f(\tilde{\gamma}), \tilde{\gamma}) ||_2}{|| C''(f(\tilde{\gamma})) ||_2} = \frac{|| [\partial_{12} C](\gamma^{*}, \tilde{\gamma})||_2}{||C''(\gamma^{*})||_2},
\end{align}
and we can rewrite our lower bound on $M$ as

\begin{align}
   M &\geq \frac{|| [\partial_{12} C](\gamma^{*}, \tilde{\gamma})||_2^2}{8 N_P^{7/2} N \overline{\text{Cov}(\hat{H}, \hat{H})} (\epsilon + \delta - \epsilon \delta)} \nonumber \\ & \left( \overline{|| D \nu(R_j ; \tilde{\gamma})||_2^2} \max_{\phi \in A} || C''(\phi)^{-1} ||_2 ||C''(\gamma^{*})||_2^2 \right)^{-1}.
\end{align}
This completes the proof.

\end{proof}

\begin{corollary}
In the case that $\hat{\gamma}$ is an unbiased estimator of $\tilde{\gamma}$, the lower bound of Theorem~\ref{thm:qcr_bnd} reduces to

\begin{align}
  M & \geq \frac{1}{8 N_P^{7/2} N \overline{\mathrm{Cov}(\hat{H}, \hat{H})} (\epsilon + \delta - \epsilon \delta)} \nonumber \\ &\left( \overline{|| D \nu(R_j ; \tilde{\gamma})||_2^2} \max_{\phi \in A} || C''(\phi)^{-1} ||_2 \right)^{-1}. 
\end{align}
\label{cor:qcr}
\end{corollary}

\begin{proof}
This follows immediately from the fact that in the unbiased case, $\gamma^{*} = \mathbb{E}[\hat{\gamma}] = \tilde{\gamma}$, so $f(\tilde{\gamma}) = \tilde{\gamma}$ and $Df(\tilde{\gamma}) = \mathbbm{1}$.
\end{proof}

We conclude this section by commenting on the scaling seen in Corollary~\ref{cor:qcr}, where we further assume that the estimator being considered is unbiased. We see that the number of pieces of evidence that we need in order to learn the parameters over a region where the fidelity loss is at most $\epsilon$ scales as $\Omega(1/\epsilon)$.  This conforms to expectations from Heisenberg-limited metrology \cite{meyer2021fisher} and therefore also to our intuitions of how the problem should scale in the best-case scenario.  Note that if the Hamiltonian terms are all Pauli operators, then the mean entrywise 1-1 norm of the covariance matrix $\text{Cov}(\hat{H}, \hat{H})$ is bounded by a constant and so we do not expect this term to dictate the scaling in most applications.  Similarly, we see that as the average norm of the Jacobian matrix of the neural network tends to zero, the lower bound diverges.  This is perhaps the most important feature of the bound as it shows that in cases where the optimization landscape is flat, a very large number of quantum states will need to be provided in order to train the generative model. Finally, we note that the lower bound depends on the smallest singular value of the Hessian $C''(\phi)$ over the region $A$. In the case that this value approaches $0$ (the curvature of the cost function vanishes), so too does the lower bound, which implies that further analysis is needed in such cases beyond the quantum Cramer-Rao bound analysis above.

It is crucial to note that the above result does not necessarily apply, precisely, to the generative model being considered in this paper. We can think of the generative model as performing parameter estimation by measuring a large data state, but Theorem~\ref{thm:qcr_bnd} also makes certain assumptions about the neural network, circuit ansatz, and optimization procedure which are in general difficult to guarantee. This includes the assumptions that the estimator yields an exact local minimum in expectation, the estimator yielded from the generative model acts as a differentiable function of the parameters $\tilde{\gamma}$, in some neighbourhood (it may rather be the case that the neural network exhibits unstable behaviour with respect to small perturbations of $\tilde{\gamma}$), and that there exist parameters $\tilde{\gamma}$ which prepare the provided quantum data in the first place. Perhaps the strongest assumption that we make is that the fidelity function is strongly convex.  In general this will not be true, but if we expand in a small enough region about the optima then the actual objective function can always be locally approximated by a strongly convex one in the vicinity of a local optima. For this reason, the assumption of strong convexity is more reasonable than it may sound, but it does preclude the general applicability. Rather, it is only applicable in the case where the estimator being considered is supported in such a region.


\section{Numerics}
\label{sec:5}

Equipped with a better theoretical understanding of the computational complexity of the proposed algorithm, as well as its fundamental limitations, we turn our attention to empirical results about its performance. More specifically, to showcase the effectiveness of the generative procedure, we utilize the model to compute approximate ground states of a collection of molecular Hamiltonians, parameterized by the nuclear coordinates of atoms in the molecule. The molecules that we consider in our simulations are H$_2$, H$_3^{+}$, H$_4$, BeH$_2$, and H$_2$O. Our goal is to create models of the molecular ground states across an entire potential energy surface from few data points. In the following subsections, we discuss some of the key ideas involved in implementation of the simulations, followed by discussion of each individual molecule, with accompanying numerical results. In the numerical examples that follow, we utilize the library PySCF~\cite{sun2018pyscf} and the minimal STO-3G basis set to perform all Hartree-Fock calculations. Among the different variants of Hartree-Fock, we use restricted Hartree-Fock (RHF) for most numerical examples, except for the case of H$_4$, where this fails to be a sufficiently high-quality initial guess (this is discussed in Sec.~\ref{sec:h2}). See Ref.\cite{arrazola2021differentiable} and Ref.\cite{szabo2012modern} for more information on the details of Hartree-Fock, basis sets, and solvers. The code used to run each of the numerical simulations is available at \url{https://github.com/XanaduAI/generative_qml_qchem}.

\subsection{State Initialization and Basis}

As was briefly discussed in Section~\ref{sec:2}, we are able to initialize the generative model in an approximation to the ground state of a molecular Hamiltonian $H(R)$, at each $R$. In particular, we set $|\psi_{\text{init}}(R)\rangle := |\psi_{\text{HF}}\rangle$ to be the qubit representation of the Hartree-Fock state at molecular coordinates $R$. In second quantization, using the Jordan-Wigner transform, each qubit is mapped to one specific spin-orbital basis function. Therefore, each bitstring basis state represents one possible electronic configuration.

For a fixed set of nuclear coordinates $R$, the Hartree-Fock method describes a family of self-consistent, numerical procedures for approximately solving the electronic Schrodinger equation, which can be run efficiently on classical computers. Hartree-Fock will yield a set of molecular orbitals whose occupation from lowest to highest energy approximates the ground state of the molecule (see Ref.\cite{szabo2012modern} for more information). The qubit representation of the Hartree-Fock state, $|\psi_{\text{HF}}\rangle$, for a given molecule at a particular geometry, is then precisely state $|11 \cdots 100 \cdots 0\rangle$ having occupation of each of the lowest-energy molecular orbitals. It follows that the qubit representation of the Hartree-Fock state, $|\psi_{\text{HF}}\rangle$, remains the same over all $R$. However, since the generated molecular orbitals change, the true molecular wavefunction represented by $|\psi_{\text{HF}}\rangle$ at distinct $R_1$ and $R_2$ can vary as well. By using the fixed $|\psi_{\text{HF}}\rangle$ at each set of coordinates $R$, we have therefore introduced an implicit assumption into our model that the basis of molecular orbitals with respect to which each qubit state corresponds, in second-quantization, varies over the nuclear coordinates $R$. This approach allows us to have good initialization of our model, without any non-trivial state preparation routines that would be required if we decided to fix a basis of molecular orbitals, and prepare the corresponding Hartree-Fock state at each individual set of nuclear coordinates.

\subsection{Circuit Ansatz}

To prepare the circuit ansatz $U$ used in the model, we perform an adaptive gate selection procedure~\cite{grimsley2019adaptive}, in which we choose operations from pools of single and double excitation gates, as defined in Ref.~\cite{arrazola2022universal}). We denote these gates as $S = \{S_1(\theta_1), \ ..., \ S_P(\theta_P)\}$ and $D = \{D_1(\phi_1), \ ..., \ D_Q(\phi_Q)\}$ respectively, based on their gradients. More specifically, for a particular set of atomic coordinates $R_0$, we utilize the following procedure to build a circuit that is able to prepare the ground state of $H(R_0)$:

\begin{enumerate}
    \item Initializing a circuit in the Hartree-Fock state, we create a parameterized circuit applying all admissible (spin-conserving) double excitation gates in $D$ to the state. Let $U_1$ be the resulting circuit, $U_1(\phi) = \prod_{j} D_j(\phi_j)$.
    \item We compute the gradient of each double excitation individually, evaluated at $\phi = 0$:
    
    \begin{equation}
        G_j = \frac{\partial}{\partial \phi_j} \langle \psi_{\text{HF}} | U_1^{\dagger}(\phi) H(R_0) U_1(\phi) |\psi_{\text{HF}} \rangle \Big|_{\phi = 0}.
    \end{equation}
    
    We construct a list $L_{D}$ of all $D_j$ such that their gradient magnitudes satisfy $|G_j| \geq \epsilon$, where we choose $\epsilon = 10^{-5}$. We then define $U_2$ to be the circuit which applies each of these double excitations, $U_2 = \prod_{D_j \in L_D} D_j$.
    \item With $U_2$, we perform VQE with respect to $H(R_0)$, which yields parameters $\phi^{*}$ that minimize the energy expectation,
    \begin{equation}
        E = \langle \psi_{\text{HF}} | U_2^{\dagger}(\phi^{*}) H(R_0) U_2(\phi^{*}) |\psi_{\text{HF}}\rangle.
    \end{equation}
    \item We create a new parameterized circuit $U_3$, this time by applying $U_2$, and then all single excitations in $S$, so $U_3(\phi, \theta) = \left[ \prod_{j} S_j(\theta_j) \right] U_2(\phi)$.
    \item Similar to Step 2, we compute gradients with respect to each single excitation individually,
    \begin{equation}
        G'_j = \frac{\partial}{\partial \theta_j} \langle \psi_{\text{HF}} | U_3^{\dagger}(\phi^{*}, \theta) | H(R_0) | U_3(\phi^{*}, \theta) | \psi_{\text{HF}} \rangle |_{\theta = 0},
    \end{equation}
    where we fix $\phi = \phi^{*}$: the optimal parameters determined in the previous step. We then construct a list $L_S$  of all $S_j$ such that $|G'_j| \geq \epsilon$. We output the combined list $L = (L_D, L_S)$ as the set of operations chosen by the adaptive procedure to prepare the ground state of $H(R_0)$.
\end{enumerate}

We then repeat these steps for a discrete set of $R_0$ over the range of atomic coordinates in which we are interested (this will be some region around $R_{\text{HF}}$). If $\{R_0,..., R_k\}$ is the set of coordinates we choose, then we obtain a sequence of lists $L_0, ..., L_k$ from performing the procedure outlined above for each of them. Taking all unique operations in each $L_j$, and combining them into one list $L_{f}$, we define our circuit ansatz to be

\begin{equation}
    U(\theta) = \displaystyle\prod_{O_j \in L_f} O_j(\theta_j).
\end{equation}

This approach allows us to tailor our ansatz to the particular set of Hamiltonians we are considering, and thus avoid unnecessarily long circuits. Empirically, we find that the set of $\{R_0, ..., R_k\}$ can be chosen very sparsely, and still capture all of the gates with non-negligible gradient with respect to any $H(R)$ in the region considered. Thus, the above procedure is likely not intractable as the dimension of the coordinate space increases.

\subsection{Classical Neural Network and Training}

The classical neural network $\nu$ used in each simulation is a simple feed-forward architecture with varying depth/width (these values are treated as hyperparameters and tuned to individual examples), implemented using the JAX library for differentiable programming \cite{jax2018github}. The sizes of the neural networks used for the simulations were modest, with the largest consisting of $3$ hidden layers of width $20$. The quantum circuit ansatz is simulated using PennyLane \cite{bergholm2018pennylane}, a library for differentiable quantum computing. The cost function $C(\gamma)$ is calculated by taking exact inner products between the states generated during the data collection phase, and the state vector outputted from the PennyLane circuit simulation. Gradients of $C$ are computed using the automatic differentiation functionality built into both JAX and PennyLane. To minimize the cost function, we utilize the Optax library \cite{optax2020github} to call an Adam optimizer implemented in JAX.

Upon termination of the optimization procedure, we compare the model's performance over a larger range of molecular coordinates outside the training data by computing fidelities between the model-produced state and the true ground state. In addition, to showcase that we can extract observable quantities from the trained model, we use our model to reconstruct ground state potential energy surfaces. Finally, we also plot the fidelity between the model-produced state and the Hartree-Fock state to demonstrate that our model is in fact learning, and is deviating far from the initial guess of $|\psi_{\text{HF}}\rangle$ over the range of molecular coordinates.

\subsection{Numerical Examples}

\begin{figure*}[ht!]
\centering
    \includegraphics[width=2.0\columnwidth]{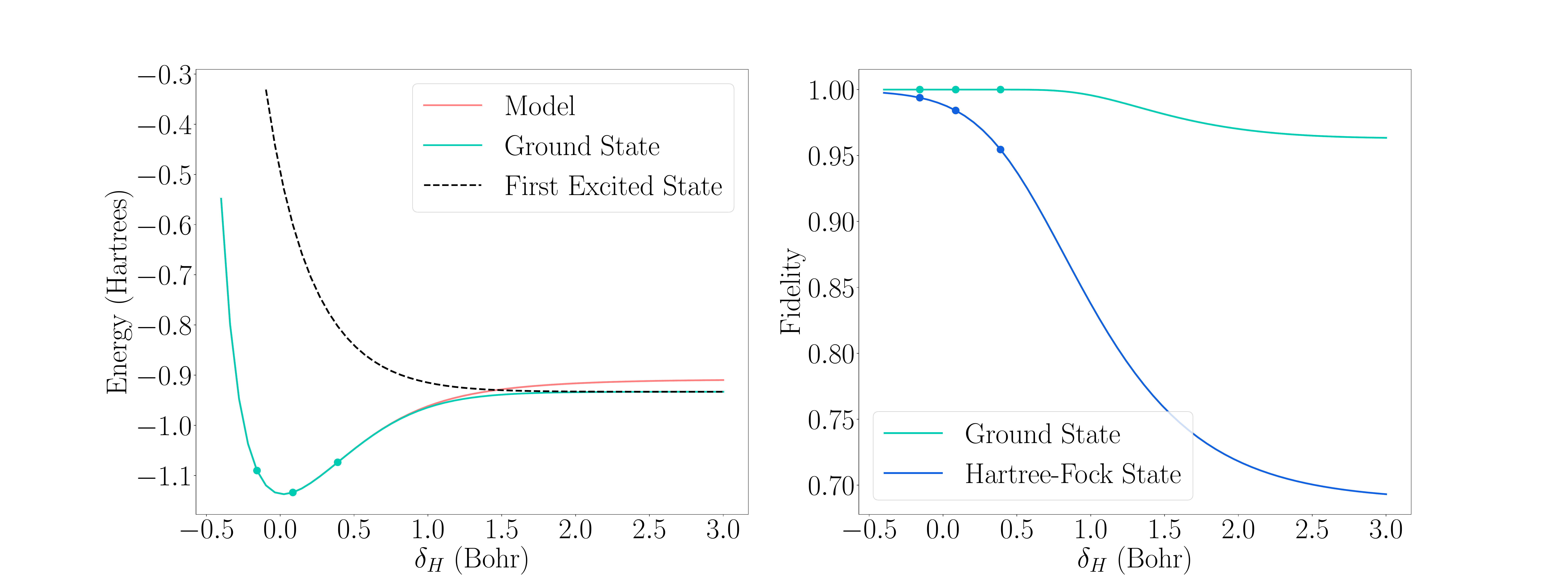}
    \includegraphics[width=2.0\columnwidth]{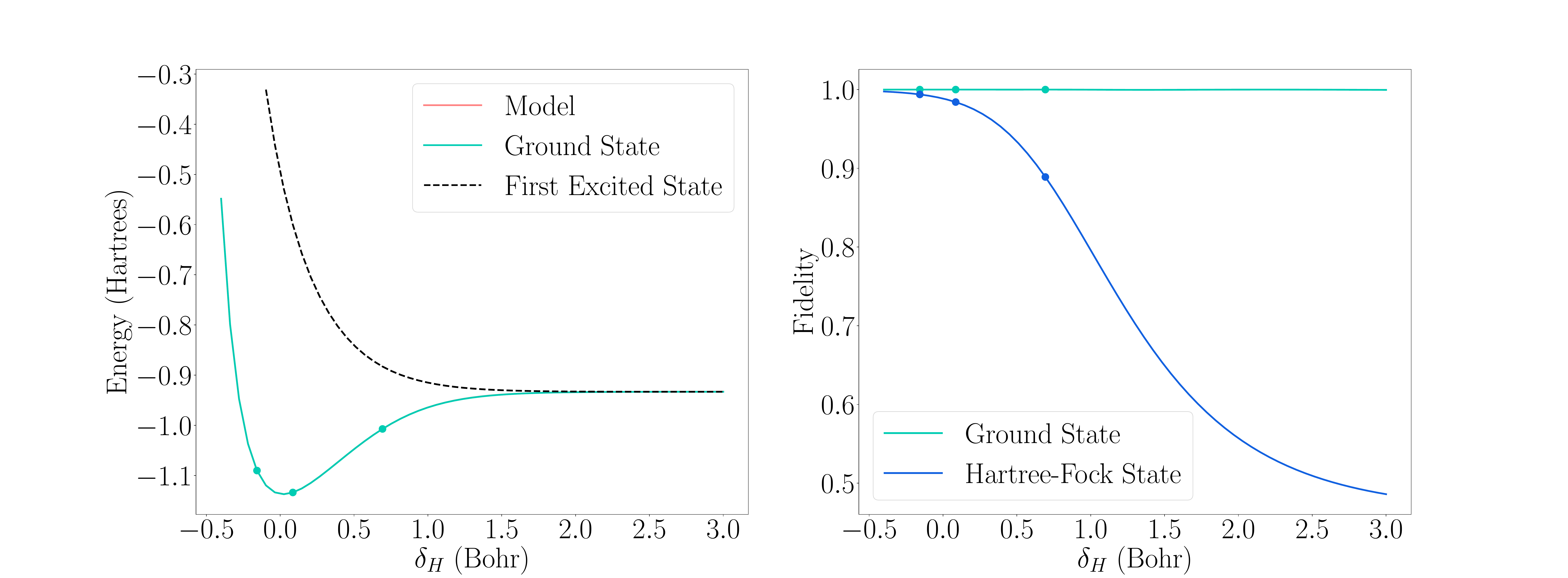}
    \caption{
    The results of the generative model for predicting the potential energy surface of H$_2$ in the STO-3G basis over the bond displacement from the equilibrium, $\delta_H$, for two different sets of training data. A total of 3 data points were used to train the model, labeled by the circular points. The top row shows one distribution of training data close to the equilibrium geometry, and the bottom row shows a distribution of training data more spread out along the bond dissociation axis. In the left column, the ground and first excited state curves, in the green line, and black dashed line respectively, are the exact result from the Hamiltonian. The the resulting curve generated from the model is shown in red. The fidelities of the model-generated ground state with respect to the exact ground state, and the Hartree-Fock state are shown in the right column with green and blue lines respectively. }
    \label{fig:h2_1}
\end{figure*}

Here we present specific numerical examples of our generative learning approach used to predict approximate groundstates for various molecular systems.  Due to computational restrictions, all examples are limited to fewer than $12$ qubits and the molecules included H$_2$, H$_3^+$, H$_4$ , BeH$_2$ and H$_2$O.  We find in all these cases that a very small number of quantum training examples need to be provided to predict the groundstates with near unit fidelity near equilibrium. 
\subsubsection{The H$_2$ and H$_4$ Molecules}
\label{sec:h2}

The first simulation we highlight is that of the H$_2$ molecule, in which we attempt to create a model $\delta \mapsto |\psi_0(R_{\text{HF}} + \delta_H e_H)\rangle$, where $R_{\text{HF}}$ is the equilibrium geometry yielded from minimizing the Hartree-Fock energy over a range of molecular coordinates, $\delta_H$ is a one-dimensional parameter, and $e_H$ is a displacement vector which equally stretches the hydrogen molecules in opposite directions along the length of their bond. Figure~\ref{fig:h2_1} shows two individual runs of the model when supplied with different sets of training data, in the top and bottom rows. The left plots show the potential energy surfaces of the exact solution of the ground and first excited state, and the resulting model-outputted energy, as a function of the bond length of H$_2$. The right plots compare fidelities of the model-produced ground state with respect to the exact ground state, and the approximate Hartree-Fock state respectively. The points on the curve show the distinct geometries at which the model is given ground state data. From Figure~\ref{fig:h2_1} we observe that the model trained on data lying farther from equilibrium demonstrates better performance: an observed phenomenon across many of the other numerical examples below, in which training data spread out over larger areas of the PES are more valuable, since it captures features of the wavefunction that are not as accessible near the equilibrium geometry. Specifically, in this case we see that in the top row of Figure~\ref{fig:h2_1}, with training data centered near the equilibrium geometry, the model is not able to accurately reproduce the energy in the dissociation limit, where the exact ground and first excited state become degenerate in energy. Subsequently, the fidelity drops in this region, indicating the the features of the wavefunction corresponding to the dissociated hydrogen atoms are not properly accounted for unless we sample farther from the equilibrium geometry. Naturally, this can be remedied by sampling closer to the dissociation limit as in the bottom row of Figure~\ref{fig:h2_1}. Specifically, when the rightmost training point is moved past the so called Coulson-Fisher point \cite{coulson1949xxxiv}, where the mean field description of the ground state of the H$_2$ molecule breaks spin symmetry, performance improves. In particular, at this point, the mean field triplet excited state wavefunction mixes with the mean field singlet ground state: a known source of difficulty when simulating bond breaking with approximate methods \cite{hait2019beyond}. Additionally, we see that from the model fidelities, there is in fact non-trivial learning in the model that improves upon the description of the mean-field Hartree-Fock ground state wavefunction (the configuration in which the model is initialized). We also note that few data points, only three in this example, are sufficient to obtain a high fidelity model across the entire potential energy curve when the training data is spread out enough, as in the bottom row of Figure~\ref{fig:h2_1}.
\begin{figure*}[t!]
    \centering
    \includegraphics[width=2.0\columnwidth]{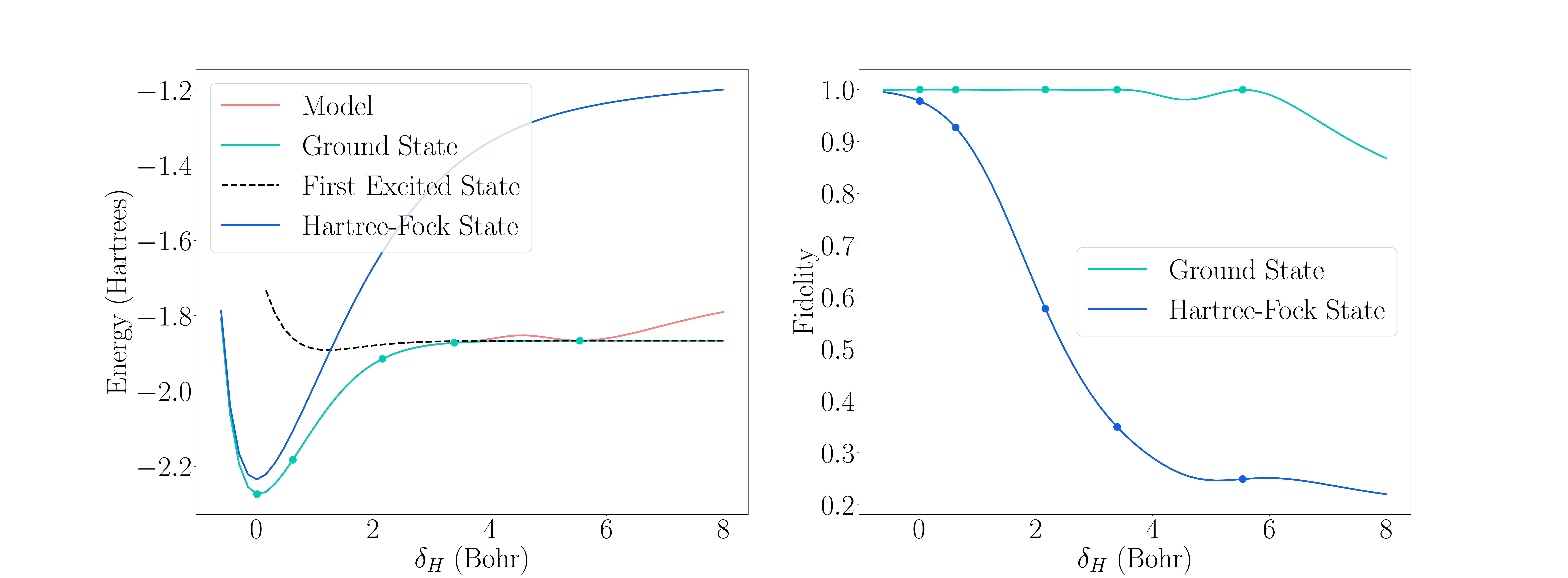}
    \caption{
    The results of the generative model for predicting the potential energy surface of rectangular H$_4$ in the STO-3G basis over the bond displacement from the equilibrium, $\delta_H$. A total of 4 data points were used to train the model, labeled by the circular points using restricted Hartree-Fock orbitals to compute the initial state of the model at each set of coordinates. On the left figure, the ground and first excited state curves, in the green line, and black dashed line respectively, are the exact result from the Hamiltonian, The Hartree-Fock state is the dark blue curve, and the resulting curve from the model is shown in red. The fidelities of the model-generated ground state with respect to the exact ground state, and the Hartree-Fock state are shown in the right column with green and blue lines respectively. The model, even when given data close to dissociation, still fails to replicate the correct behaviour of the target state when RHF orbitals are used.}
    \label{fig:h4_bad}
\end{figure*}

Next, we investigate the symmetrically stretched rectangular H$_4$ model. H$_4$ is often times used as a model for strongly correlated systems in chemistry, since over most of the bonding coordinate, including the equilibrium geometry, the Hartree-Fock mean-field wavefunction has low fidelity with respect to the exact wavefunction. Physically, this means that there are multiple electronic configurations non-trivially contributing to the overall exact ground state wavefunction, making this a harder challenge for our model to generate the ground states from a Hartree-Fock reference. For this system we use a model of the form $\delta \mapsto |\psi_0(R_{\text{HF}} + \delta_H e_H)\rangle$ in which $R_{\text{HF}}$ is the rectangular Hartree-Fock equilibrium configuration of the molecule, $\delta_H$ is again a one-dimensional parameter, and $e_H$ is a displacement vector which stretches the sides of the rectangle formed by the hydrogen atoms equally.

\begin{figure*}[t!]
    \centering
    \includegraphics[width=2.0\columnwidth]{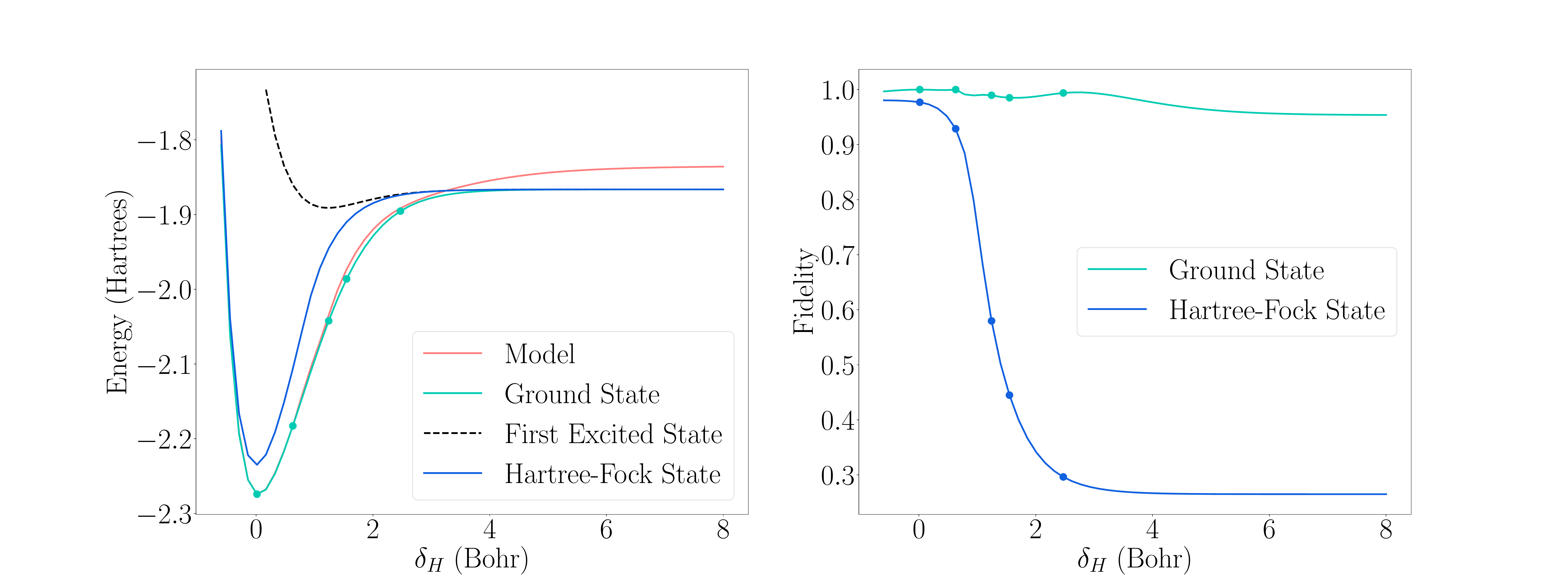}
    \includegraphics[width=2.0\columnwidth]{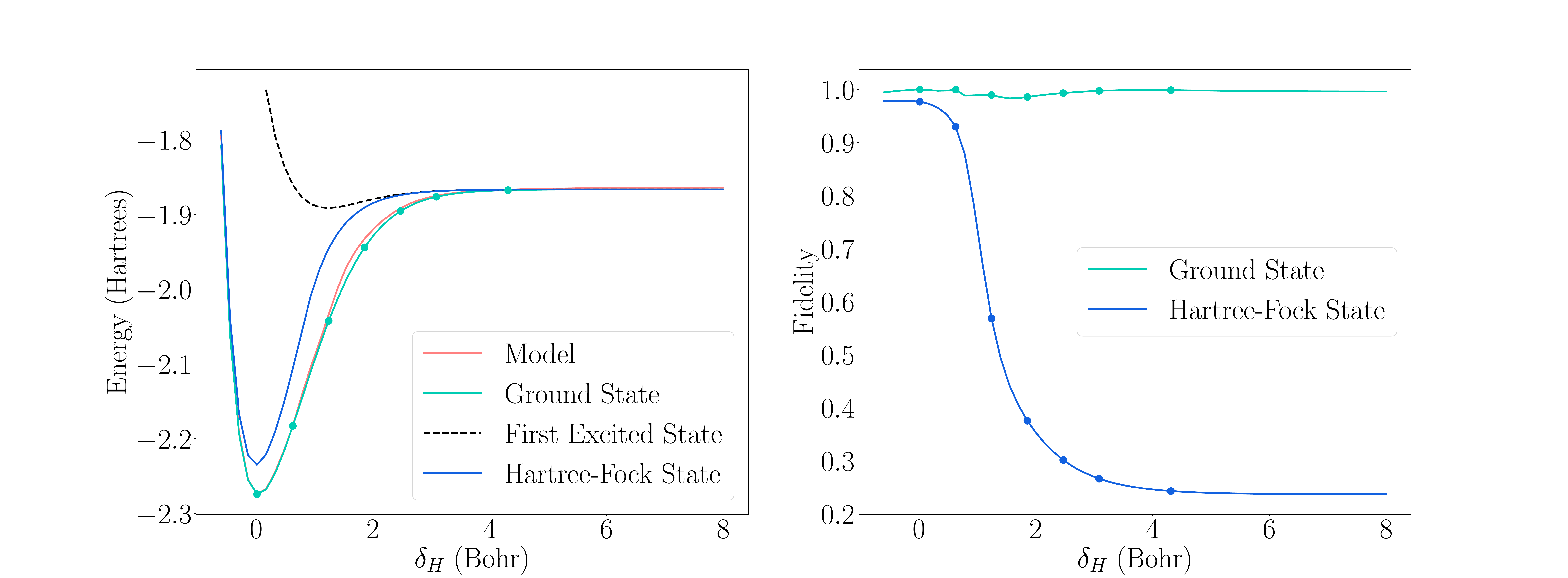}
    \caption{
    The results of the generative model for predicting the potential energy surface of H$_4$ in the STO-3G basis over the bond displacement from the equilibrium, $\delta_H$, for two different sets of training data. A total of 5 data points were used to train the model, labeled by the circular points. The top row shows one distribution of training data closer to the equilibrium geometry, and the bottom row shows a distribution of training data more spread out along the bond dissociation axis. In the left column, the ground and first excited state curves, in the green line, and black dashed line respectively, are the exact result from the Hamiltonian. The dark blue line is the Hartree-Fock ground state, and the resulting curve from the model is shown in red. The fidelities of the model-generated ground state with respect to the exact ground state, and the Hartree-Fock state are shown in the right column with green and blue lines respectively. Note that performance improves drastically over the cases when RHF is used, as in Figure~\ref{fig:h4_bad}}
    \label{fig:h4_better}
\end{figure*}

The first run, presented in Figure~\ref{fig:h4_bad}, shows the model result trained on data using restricted Hartree-Fock orbitals. While certain parts of the model reproduce the exact results well, there are un-physical deviations in the energy, that do not level out in dissociation limit. Even when given data approaching the dissociation limit, the model still fails to replicate the true behaviour of the system, as is indicated by both the increase of energy upon dissociation, and the low fidelity between the model and the true ground state.
While this error could be the result of a bad generative model, the origin of the issue arises from the fact that using restricted Hartree-Fock orbitals fails to be a sufficiently robust initial guess for the ground state of the system. This can partly be attributed to the fact that the RHF ground state does not reproduce the correct dissociation limit in H$_4$ as shown in this dark blue curve in Figure~\ref{fig:h4_bad}. By utilizing \textit{unrestricted} Hartree-Fock (UHF) orbitals, where electrons of different spins are not assumed to be paired in the same spatial orbital, the correct molecular energy in the limit of bond dissociation can be reproduced.~\cite{szabo2012modern}. By using UHF orbitals, we find that the model produces a more reasonable PES, and the generated ground states are of similar quality to H$_2$. The data for the UHF orbital based model is presented in Figure~\ref{fig:h4_better} for two different training sets of 5 points each, in the top row and bottom row respectively where the bottom row has more training data near the dissociation limit, reproducing a high quality PES compared to the exact result.

We do observe, however, in each of these examples, small fluctuations in the fidelity of the model-produced state with the true ground state between $\delta_{H} = 0.5 \ \text{Bohr}$ and $1 \ \text{Bohr}$. We hypothesize that this is an artifact of Hartree-Fock, in which the self-consistent procedure yields discontinuities when performing calculations at different values of $R$. Evidence for this phenomenon is provided in Figure~\ref{fig:hf_bad}, where we observe a bifurcation in entries of the molecular orbital coefficient matrix from unrestricted Hartree-Fock (the matrix that determines the molecular orbitals) in the region between $0.5$ and $1$ Bohr. Addressing this issue seems to be a non-trivial task, and is likely a promising area of future study.

\begin{figure*}[t!]
    \centering
    \includegraphics[width=2.0\columnwidth]{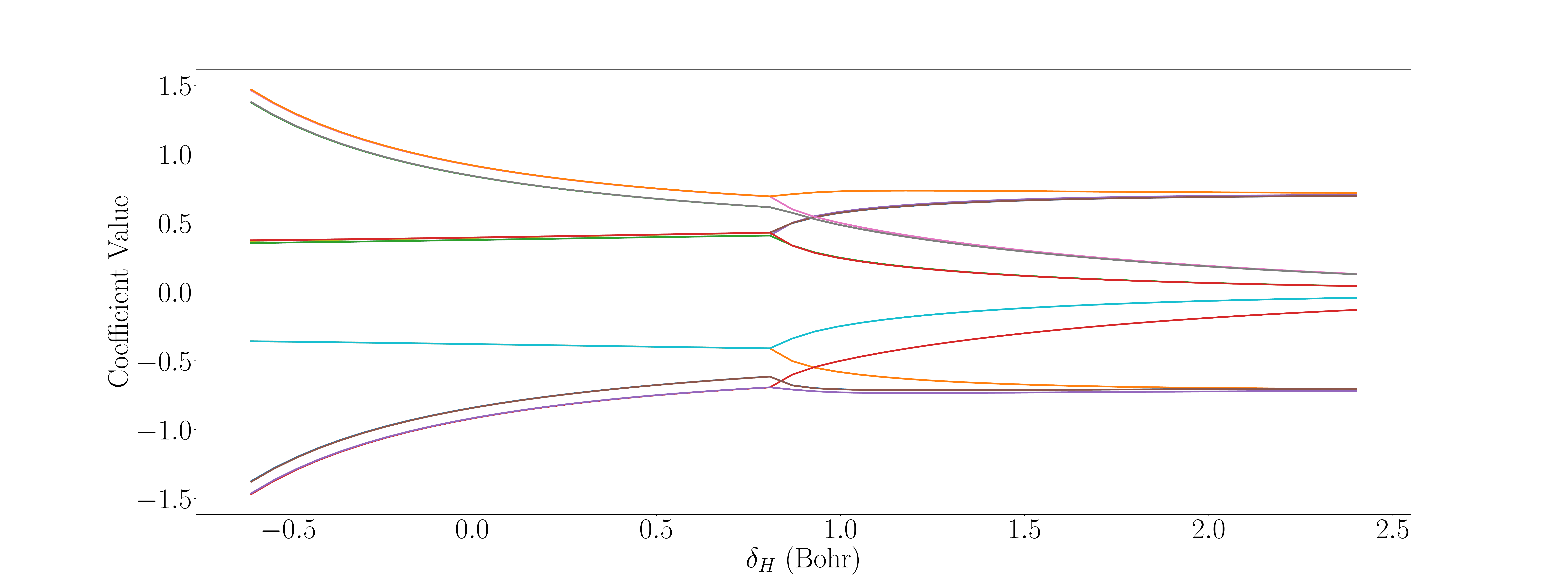}
    \caption{A plot showing the molecular orbital coefficient matrix entries from the unrestricted Hartree-Fock calculation for the H$_4$ molecule using the STO-3G basis set, with respect to the rectangular bond displacement $\delta_H$. As the plot shows, there is a clear bifurcation in the matrix elements between $0.5$ Bohr and $1.0$ Bohr, which we expect contributes to the small dip in the fidelity of the model-produced state near this point.}
    \label{fig:hf_bad}
\end{figure*}

\subsubsection{The H$_3^{+}$ Molecule}

To showcase the ability of the proposed model to generate ground states from multi-dimensional parameters specifying molecular geometries, we consider the H$_3^{+}$ molecule in a planar triangular formation. We fix one of the hydrogen atoms in place, and vary the angle between the bonds of the remaining two hydrogen atoms with the fixed one, as well as the bond lengths. We also constrain the bond lengths from the fixed hydrogen atom to the other two hydrogen atoms to be the same length. Specifically, this causes the geometry of the atoms to match that of an isosceles triangle. It then follows that the model we construct is a two-dimensional model of the parameters $(\theta_H, \delta_H)$, where $\theta_H$ is the bond angle, and $\delta_H$ is the bond length from the fixed atom.

Similar to H$_2$ and H$_4$ single paramter PES, our model is able to learn the true ground state with high fidelity across a wide range of geometries for the 2-dimensional PES of H$_3^{+}$ with limited data. When more data is provided, and scattered farther among the region of coordinates considered, the model performs better, as expected. In Figures~\ref{fig:h3_plus_1} and \ref{fig:h3_plus_2}, we showcase contour plots comparing ground-state energies and fidelities over the two-dimensional parameter space, for both the cases of fewer and more data points (6 and 10 ground states in the training set respectively). 

\begin{figure*}
    \centering
    \includegraphics[width=2.0\columnwidth]{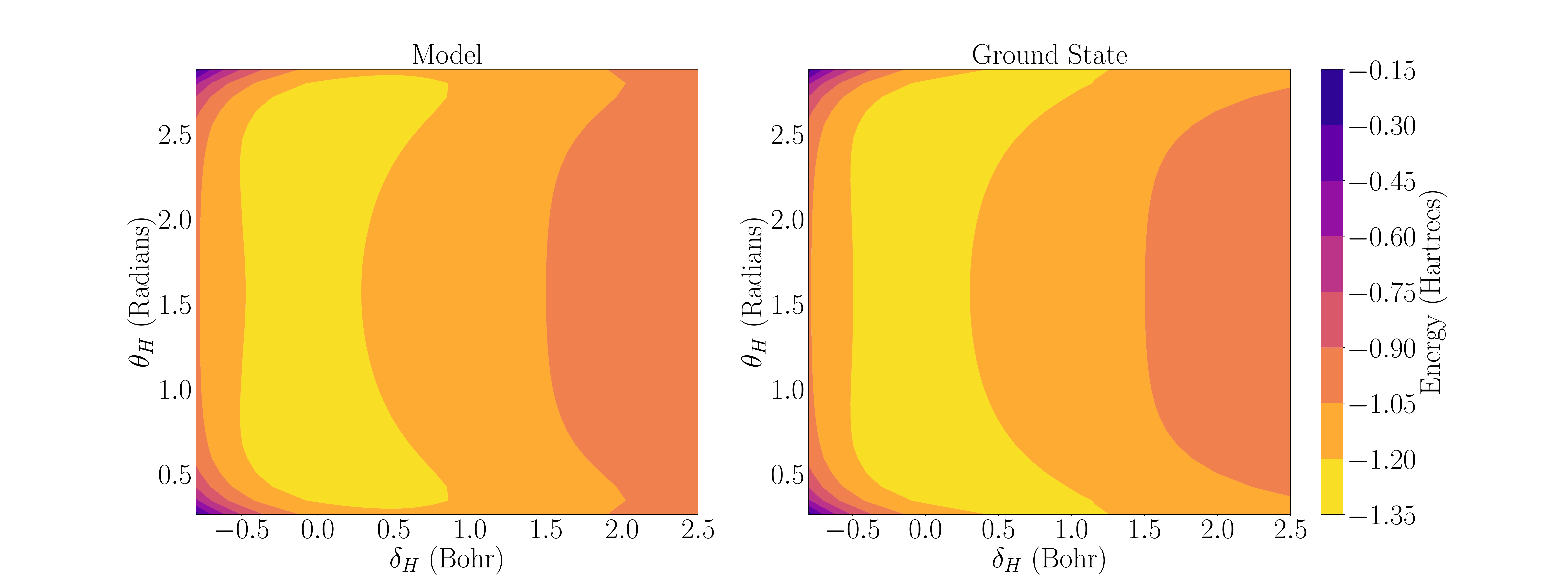}
    \includegraphics[width=1.4\columnwidth]{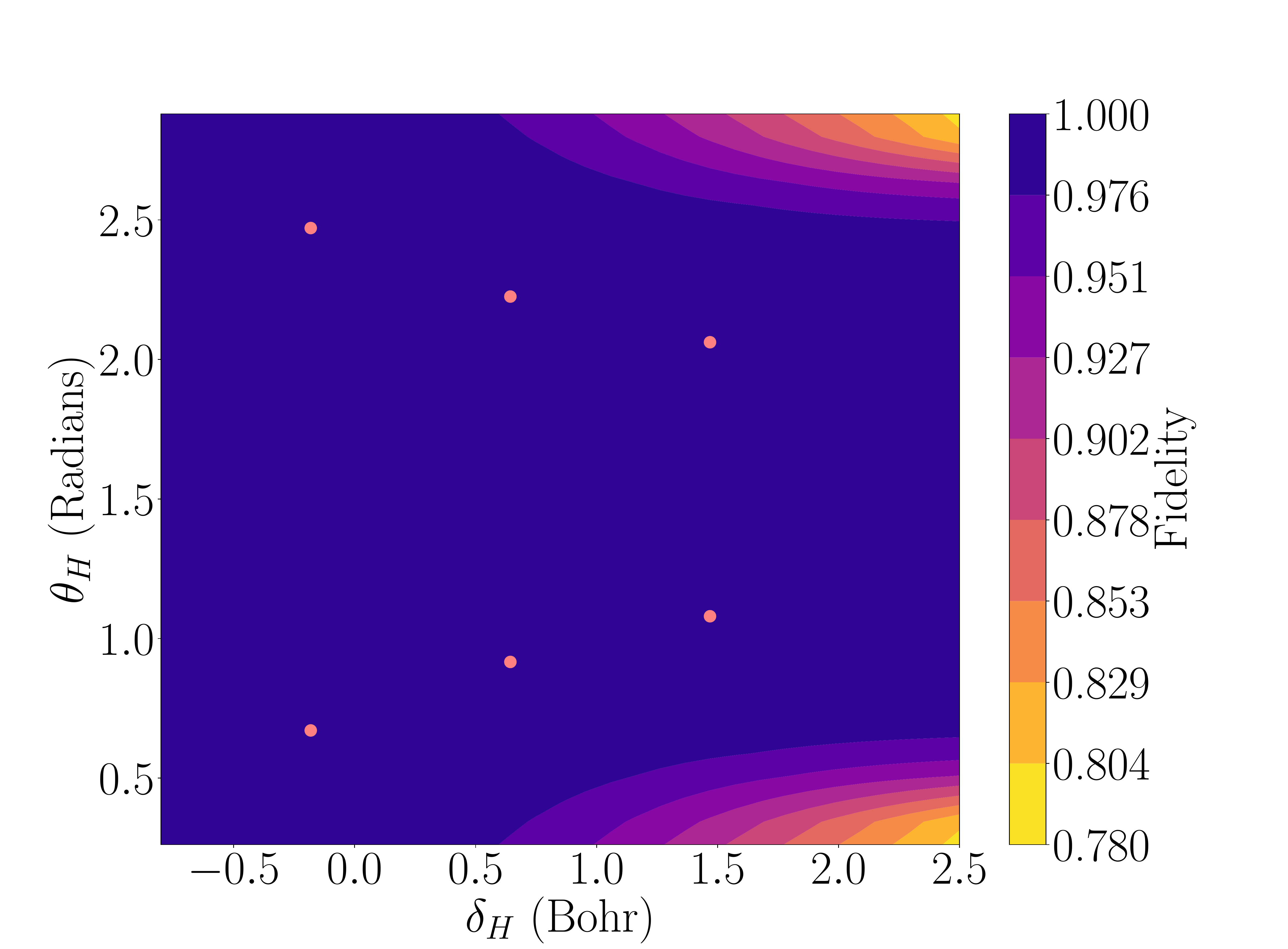}
    \caption{ 
    The results of the generative model for predicting the potential energy surface of triangular H$_3^+$ in the STO-3G basis as a function of the separation between the central hydrogen and the two remaining hydrogens, $\delta_H$, and the bond angle between them, $\theta_H$. The top plots compare the two-dimensional potential energy surfaces of the model produced state (left) and the target state (right), over the parameter space. The bottom plot shows the fidelity of the model state with the target state,  where the 6 training data points are labeled with the orange circles.}
    \label{fig:h3_plus_1}
\end{figure*}

\begin{figure*}
    \centering
    \includegraphics[width=2.0\columnwidth]{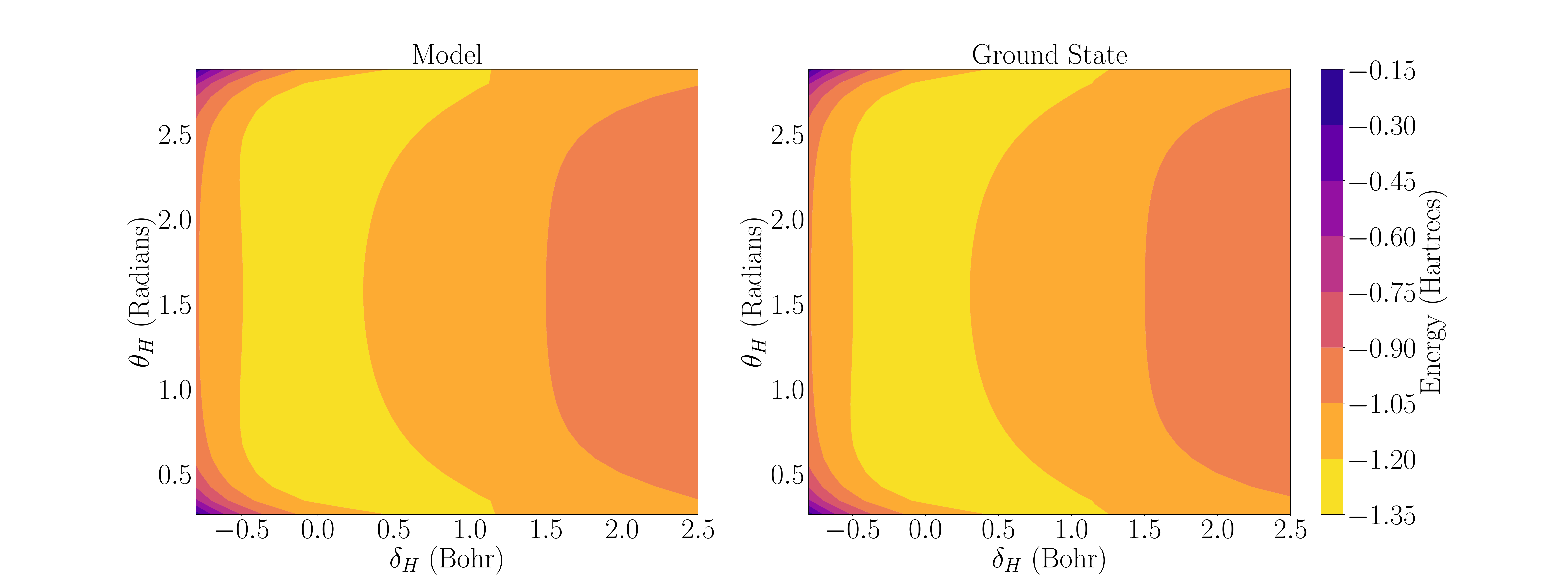}
    \includegraphics[width=1.4\columnwidth]{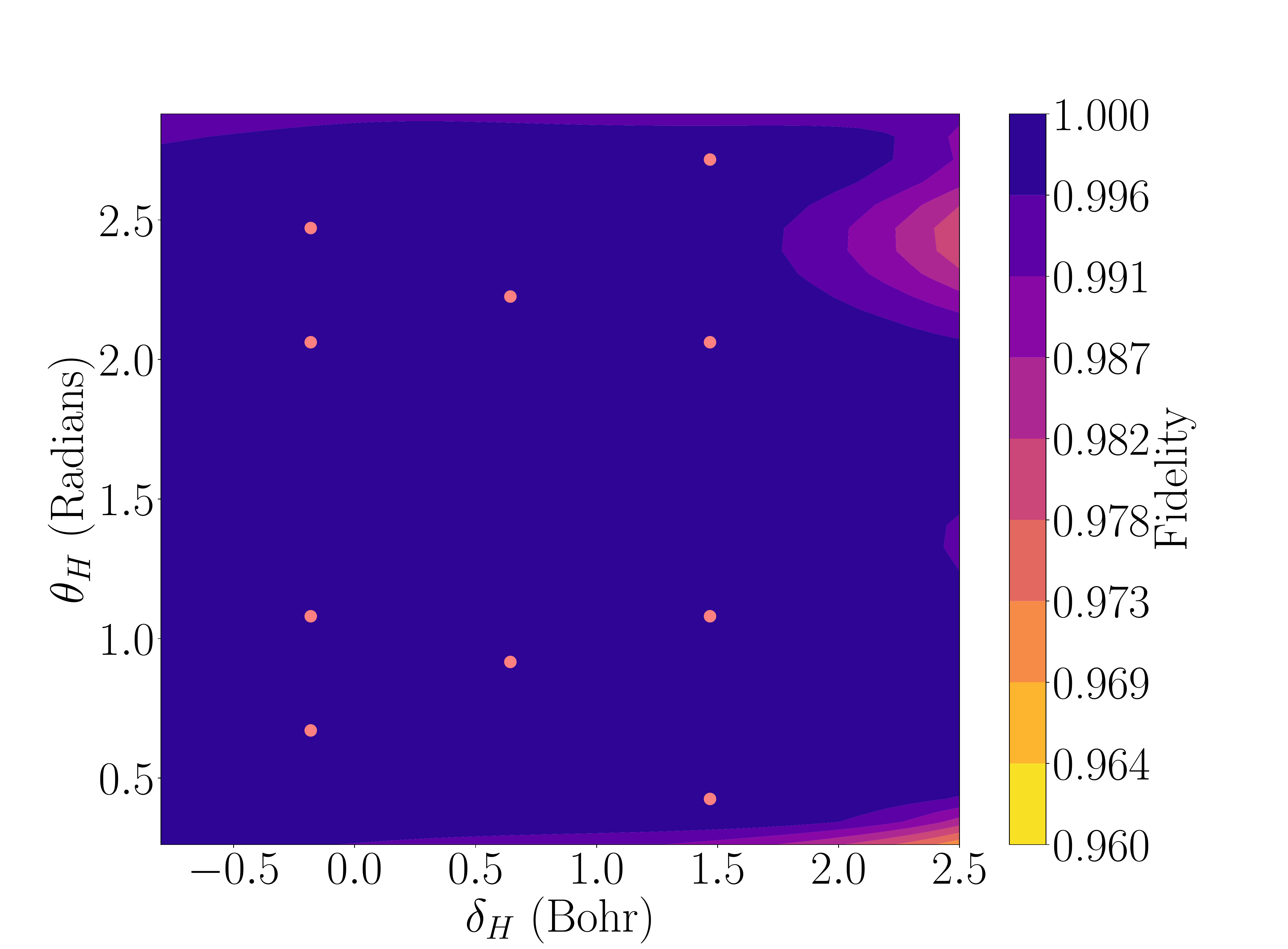}
    \caption{ 
    The results of the generative model for predicting the potential energy surface of triangular H$_3^+$ in the STO-3G basis as a function of the separation between the central hydrogen and the two remaining hydrogens, $\delta_H$, and the bond angle between them, $\theta_H$. The top plots compare the two-dimensional potential energy surfaces of the model produced state (left) and the target state (right), over the parameter space. The bottom plot shows the fidelity of the model state with the target state, where the 10 training data points are labeled with the orange circles. With more data, the quality of the model improves compared to the instance in Figure~\ref{fig:h3_plus_1}}
    \label{fig:h3_plus_2}
\end{figure*}

\subsubsection{The BeH$_2$ and H$_2$O Molecules}

\begin{figure*}
    \centering
    \includegraphics[width=2.0\columnwidth]{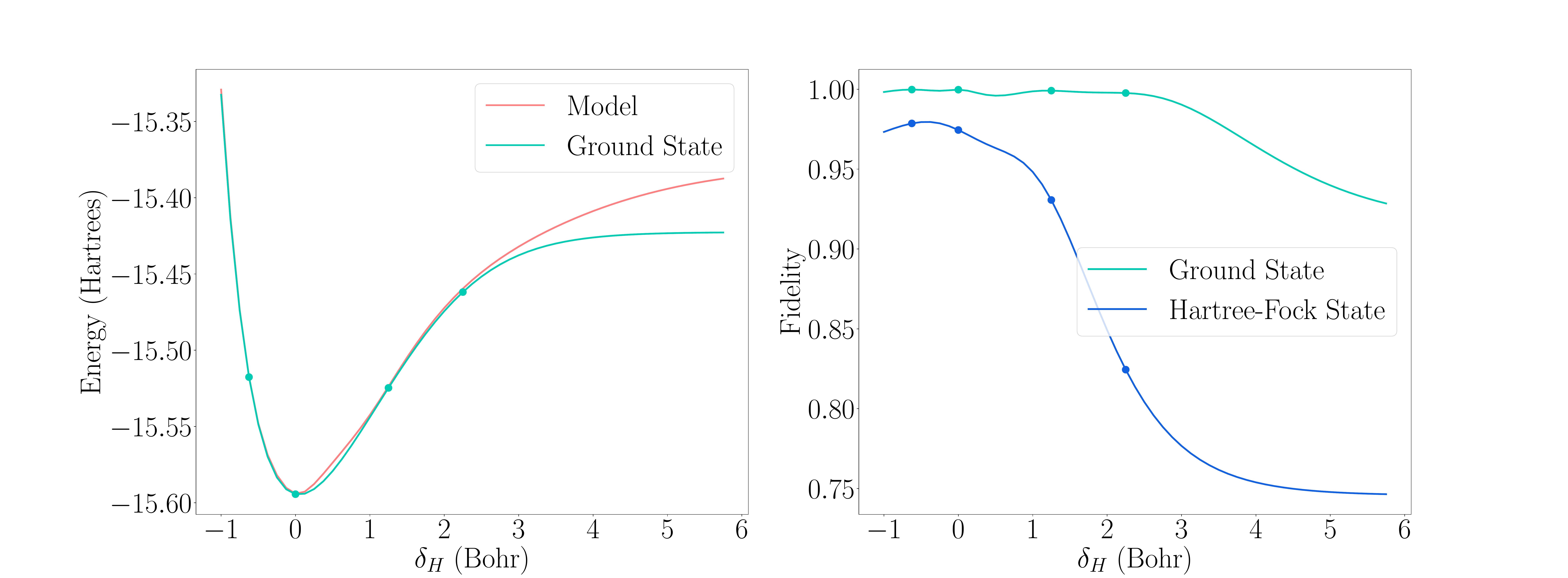}
    \includegraphics[width=2.0\columnwidth]{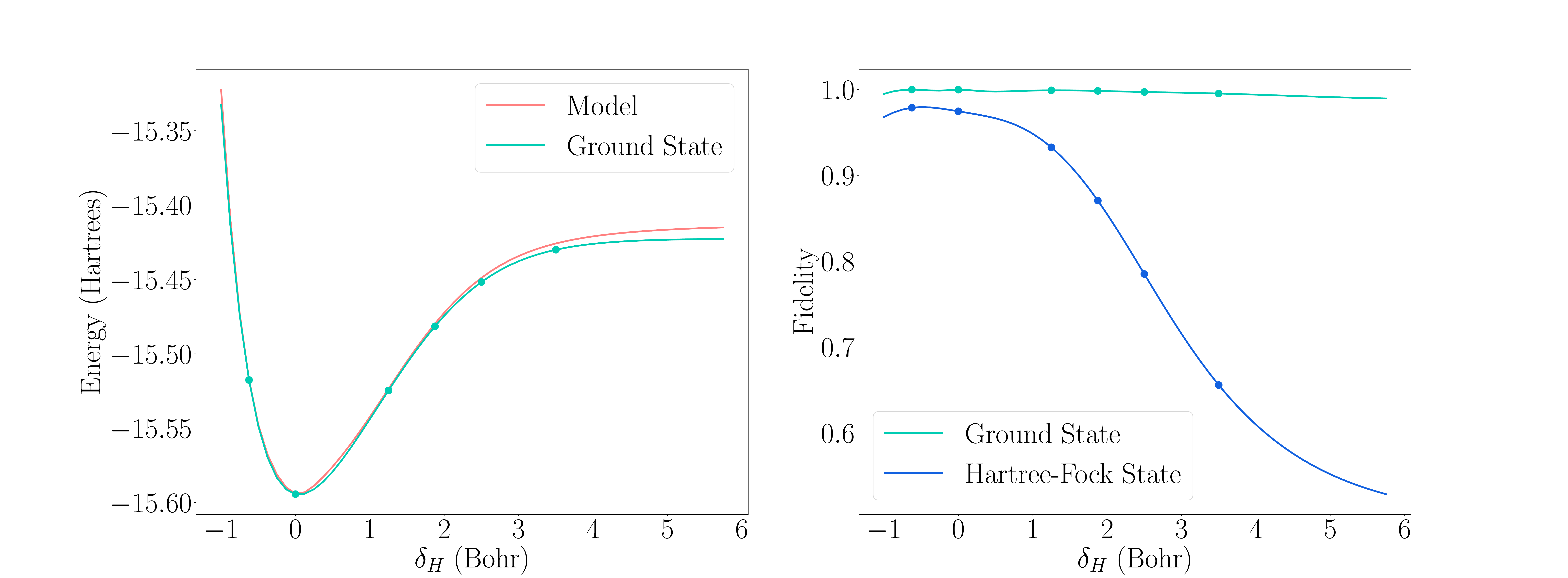}
    \caption{
    The results of the generative model for predicting the potential energy surface of BeH$_2$ in the STO-3G basis, having one of its hydrogen atoms stretched away from the beryllium atom, along its bond over the bond displacement from the equilibrium, $\delta_H$, for two different sets of training data. A total of 4 (top row), and 6 (bottom row) data points were used to train the model, labeled by the circular points. The top row shows one distribution of training data closer to the equilibrium geometry, and the bottom row shows a distribution of training data more spread out along the bond dissociation axis. In the left column, the exact ground state curve is shown in the green line and the resulting curve from the model is shown in red. The fidelities of the model-generated ground state with respect to the exact ground state, and the Hartree-Fock state are shown in the right column with green and blue lines respectively.}
    \label{fig:beh2}
\end{figure*}

\begin{figure*}
    \centering
    \includegraphics[width=2.0\columnwidth]{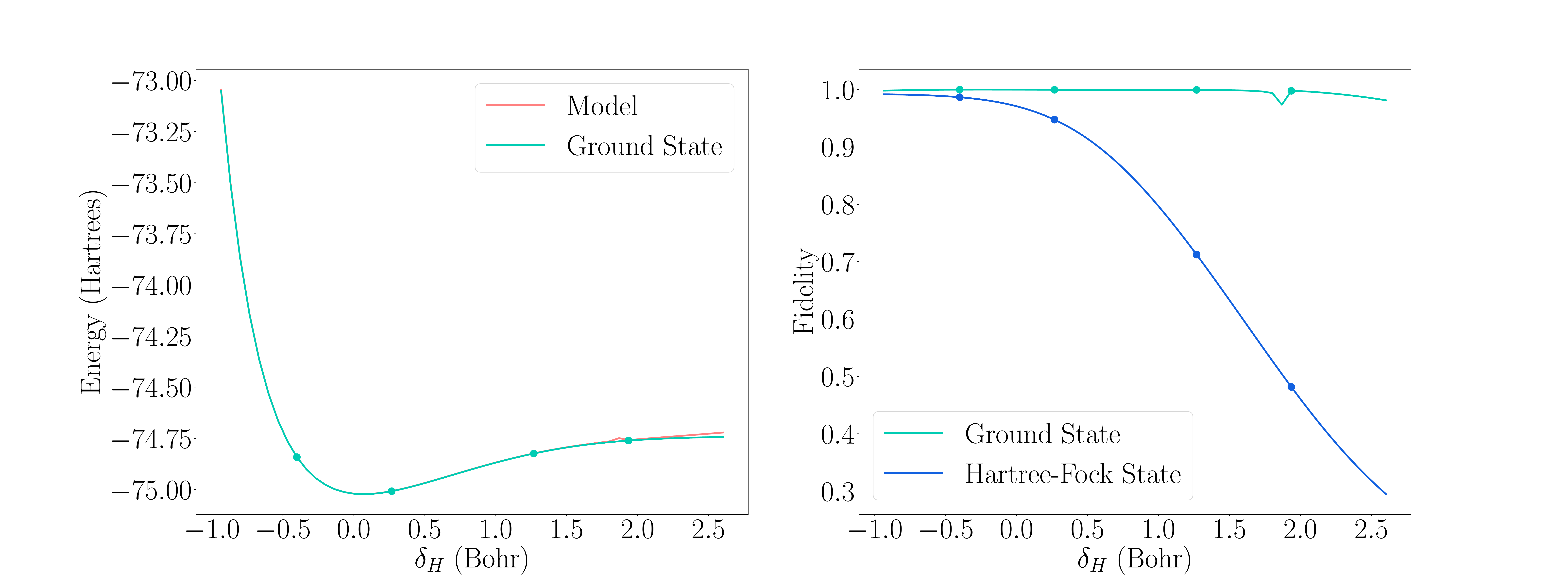}
    \caption{
    The results of the generative model for predicting the potential energy surface of H$_2$O in the STO-3G basis, having both of its hydrogen atoms symmetrically stretched away from the oxygen atom along their respective bonds, with respect to the bond displacement from equilibrium, $\delta_H$, for one set of training data with 4 data points, labeled by the circular points. In the left figure, the exact ground state curve is shown in the green line and the resulting curve from the model is shown in red. The fidelities of the model-generated ground state with respect to the exact ground state, and the Hartree-Fock state are shown in the right figure, with green and blue lines respectively.}
    \label{fig:h2o}
\end{figure*}

To conclude the numerical demonstrations presented in this paper, we highlight examples of larger molecules with more complicated electronic structures and molecular orbital bases, and show that the generative model is still able to learn accurate representations of the corresponding ground states. In the case of BeH$_2$, we consider the one-dimensional problem of the molecule in its linear configuration, and let $\delta_H$ correspond to a perturbation of only one hydrogen atom from its equilibrium configuration, along the length of its bond with the beryllium atom. For the case of water, we let $\delta_{H}$ denote the symmetric stretching of both hydrogen atoms along their bonds with oxygen. Results of these simulations are given in Figures~\ref{fig:beh2} and \ref{fig:h2o}, and once again show good performance, as long as we use a sufficiently large set of data. In the case of water, we once again see a small, unexpected fluctuation in the ground state energy and fidelity near $\delta_{H} = 1.8 \ \text{Bohr}$. As is the case for H$_4$, this small dip may be an artifact of Hartree-Fock but further numerical investigations will be required to verify this fact.

\section{Conclusion}
We have presented a study into how generative methods can be used in the context of quantum chemistry, specifically for preparing ground states of a parameterized molecular Hamiltonian based on its geometry.  Unlike existing classical approaches, the input data to the considered model is fully quantum and the output is a quantum state. The central idea behind our approach is to combine a classical neural network that processes the nuclear positions and translates them into parameters of a quantum circuit. The loss function is chosen to be the infidelity between the output state generated and the training state at each location in the training set.  Given an incoherent oracle that provides copies of the training data, we show that the number of training samples needed to estimate the gradient within error $\epsilon$ with probability $1-\delta$ scales as $O(N\log(1/\delta)/\epsilon^2)$
and in the case where coherent access to the training vectors is provided the scaling with $\epsilon$ can be quadratically improved.

We further provide arguments that all such algorithms must face obstacles when trying to learn across phase transitions by applying a reduction from Grover's search to the fermionic ground-state prediction problem and also place fundamental limitations on the amount of data required to construct a good estimate of the model's optimal neural network parameters, for an idealized version of the described procedure.

We test these ideas on several benchmark molecules including hydrogen-only molecules, BeH$_2$ and H$_2$O.  We observe that even given a small number of training examples, we can construct a simple quantum neural network that is capable of generating the approximate groundstate of the molecules with fidelity greater than $99\%$ in most of the cases considered. Further, once trained, the parameterized circuits generated have far less cost than applying phase estimation to project onto the state in question.  This shows that using classical and quantum neural networks in tandem provides a promising way to prepare the ground state of molecules over an entire potential energy surface and may potentially be a significant application for quantum machine learning.

There are several further avenues of inquiry raised by this work.  First, there is the question of scalability of the method and question about the existence of barren plateaus in the optimization landscape.  While the numerical experiments provided give evidence that the presented ideas work in practice, at least in some test cases, there are potential issues due to vanishing gradients arising during training.  While this matter is unlikely to be fully addressed until scalable quantum computers are available, subsequent numerical and theoretical studies are needed to shed light on the issue. A secondary question involves whether or not these methods can be practically generalized to generating time dynamics.  It is conceivable that the approaches outlined in this paper could be extended to predict quantum states as a function of time, rather than simply predicting them as a function of nuclear position. This would open up our concept to broader families of problems within quantum simulation. A final question that remains open by this work is the choice of distribution of training examples.  At present, while our lower bound results provide some understanding of how the prediction accuracy must scale with the number and location of the training examples it does not provide a prescription for choosing their positions.  The development of an online training algorithm that can identify the most profitable training vectors to include in a corpus to improve the predictive power of the model would be a substantial extension in this setting since ground states of gapped Hamiltonians can often be prepared in polynomial time by a quantum computer and thus training data can often be provided over a wide parameter range at low cost.  The optimization of training points, along with finding improved heuristics for the quantum neural networks may ultimately lead us to methods that will finally provide practical applications of quantum machine learning.

\section{Acknowledgements}
We thank Jarrod McClean, Alba Cervera-Lierta and Jakob S. Kottmann for their constructive feedback for this manuscript.
J.C. acknowledges partial funding from a MITACS Accelerate Grant. C.O.M. acknowledges funding by Laboratory Directed Research and Development Program and Mathematics for Artificial Reasoning for Scientific Discovery investment at the Pacific Northwest National Laboratory, a multiprogram national laboratory operated by Battelle for the U.S. Department of Energy under Contract DE-AC05- 76RLO1830. T.F.S. is a Quantum Postdoctoral Fellow at the Simons Institute for the Theory of Computing, supported by the U.S. Department of Energy, Office of Science, National Quantum Information Science Research Centers, Quantum Systems Accelerator.
MK was supported by the Sydney Quantum Academy, Sydney, NSW, Australia and by the Defense Advanced Research Projects Agency under Contract No. HR001122C0074. Any opinions, findings and conclusions or recommendations expressed in this material are those of the author(s) and do not necessarily reflect the views of the Defense Advanced Research Projects Agency.  MK and NW conducted part of the research at KITP supported in part by the National Science Foundation under Grant No. NSF PHY-1748958. MK is a member of Google Quantum AI but this manuscript consists exclusively of her work at UTS and it is not a product of Google Research.
NW was also funded by grants from the US Department of Energy, Office of Science, National Quantum Information Science Research Centers, Co-Design Center for Quantum Advantage under contract number DE-SC0012704.

\bibliographystyle{unsrt}
\bibliography{ref}

\appendix

\end{document}